\documentclass[letterpaper,11pt]{article}

\usepackage{amsmath,amsfonts,amsthm,amssymb,color}
\usepackage{fancybox}
\usepackage[colorlinks=true]{hyperref}
\usepackage{float}
\usepackage{setspace}
\usepackage{fullpage}
\usepackage{graphicx}

\usepackage{fancybox}
\usepackage{booktabs}
\usepackage{multirow}
\usepackage{threeparttable}
\usepackage{makecell}

\newtheorem{problem}{Problem}
\newtheorem{theorem}{Theorem}[section]

\newtheorem{lemma}[theorem]{Lemma}

\newtheorem{fact}[theorem]{Fact}

\theoremstyle{definition}
\newtheorem{definition}[theorem]{Definition}
\newtheorem{remark}[theorem]{Remark}

\newenvironment{fminipage}%
  {\begin{Sbox}\begin{minipage}}%
  {\end{minipage}\end{Sbox}\fbox{\TheSbox}}

\def\pleq{\preccurlyeq}

\def\defeq{\stackrel{\mathrm{def}}{=}}
\def\setof#1{\left\{#1  \right\}}
\def\sizeof#1{\left|#1  \right|}

\def\ceil#1{\left\lceil #1 \right\rceil}

\def\union{\cup}

\def\eps{\epsilon}

\def\abs#1{\left|#1  \right|}

\def\trace#1{\mathrm{Tr} \left(#1 \right)}
\def\norm#1{\left\| #1 \right\|}

\newcommand\PPi{\boldsymbol{\Pi}}

\newcommand\er{R_{\mathrm{eff}}}

\def\aa{\pmb{\mathit{a}}}
\newcommand\bb{\boldsymbol{\mathit{b}}}
\newcommand\cc{\boldsymbol{\mathit{c}}}

\newcommand\ee{\boldsymbol{\mathit{e}}}

\newcommand\ii{\boldsymbol{\mathit{i}}}

\newcommand\vv{\boldsymbol{\mathit{v}}}

\newcommand\xx{\boldsymbol{\mathit{x}}}

\renewcommand\AA{\boldsymbol{\mathit{A}}}
\newcommand\BB{\boldsymbol{\mathit{B}}}
\newcommand\CC{\boldsymbol{\mathit{C}}}

\newcommand\II{\boldsymbol{\mathit{I}}}

\newcommand\MM{\boldsymbol{\mathit{M}}}
\newcommand\LL{\boldsymbol{\mathit{L}}}
\newcommand\HH{\boldsymbol{\mathit{H}}}
\newcommand\PP{\boldsymbol{\mathit{P}}}
\newcommand\QQ{\boldsymbol{\mathit{Q}}}
\newcommand\RR{\boldsymbol{\mathit{R}}}
\renewcommand\SS{\boldsymbol{\mathit{S}}}
\newcommand\TT{\boldsymbol{\mathit{T}}}

\newcommand\WW{\boldsymbol{\mathit{W}}}

\newcommand\XX{\boldsymbol{\mathit{X}}}

\newcommand\ZZ{\boldsymbol{\mathit{Z}}}

\newcommand\ZZtil{\widetilde{\boldsymbol{\mathit{Z}}}}

\newcommand\Otil{\widetilde{O}}

\newcommand{\zero}{\mathbf{0}}
\newcommand{\one}{\mathbf{1}}

\DeclareMathOperator*{\argmin}{arg\,min}
\DeclareMathOperator*{\argmax}{arg\,max}

\newcommand{\kh}[1]{\left(#1\right)}

\newcommand{\exactGreedy}{\textsc{ExactGreedy}}
\newcommand{\FastGreedy}{\textsc{ApproxGreedy}}
\newcommand{\SDDMSolver}{\textsc{Solve}}
\newcommand{\ERSumEst}{\textsc{ERSumsEst}}
\newcommand{\GainsEst}{\textsc{GainsEst}}

\newcommand{\lambdamin}{\lambda_{\mathrm{min}}}
\newcommand{\lambdamax}{\lambda_{\mathrm{max}}}
\newcommand{\wmax}{w_{\mathrm{max}}}
\newcommand{\wmin}{w_{\mathrm{min}}}

\newenvironment{algbox}[0]{\vskip 0.2in
\noindent
\begin{center}
\begin{fminipage}{0.85\textwidth}
}{
\end{fminipage}
\end{center}
\vskip 0.2in
}

\title{Current Flow Group Closeness Centrality for Complex Networks}

\author{
	Huan Li \\
	\normalsize School of Computer Science \\
	\normalsize Fudan University\\
	\normalsize \texttt{huanli16@fudan.edu.cn}
	\and
	Richard Peng \\
	\normalsize School of Computer Science \\
	\normalsize Georgia Institute of Technology\\
	\normalsize \texttt{rpeng@cc.gatech.edu}
	\and
	Liren Shan \\
	\normalsize School of Computer Science \\
	\normalsize Fudan University\\
	\normalsize \texttt{13307130150@fudan.edu.cn}
	\and
	Yuhao Yi \\
	\normalsize School of Computer Science \\
	\normalsize Fudan University\\
	\normalsize \texttt{yhyi15@fudan.edu.cn}
	\and
	Zhongzhi Zhang \\
	\normalsize School of Computer Science \\
	\normalsize Fudan University\\
	\normalsize \texttt{zhangzz@fudan.edu.cn}
}

\date{}

\begin{document}

{\singlespacing
\maketitle

\begin{abstract}
Current flow closeness centrality (CFCC) has a better discriminating ability than the ordinary closeness centrality based on shortest paths. In this paper, we extend this notion to a group of vertices in a weighted graph, and then study the problem of finding a subset $S$ of $k$ vertices to maximize its CFCC $C(S)$, both theoretically and experimentally. We show that the problem is NP-hard, but propose two greedy algorithms for minimizing the reciprocal of $C(S)$ with provable guarantees using the monotoncity and supermodularity. The first is a deterministic algorithm with an approximation factor $(1-\frac{k}{k-1}\cdot\frac{1}{e})$ and cubic running time; while the second is a randomized algorithm with a $(1-\frac{k}{k-1}\cdot\frac{1}{e}-\eps)$-approximation and nearly-linear running time for any $\epsilon > 0$.
Extensive experiments on model and real networks demonstrate that our algorithms are effective and efficient, with the second algorithm being scalable to massive networks with more than a million vertices.
\end{abstract}%

\section{Introduction}


A fundamental problem in network science and graph mining is to identify crucial vertices~\cite{LaMe12,LuChReZhZhZh16}. It is an important tool in network analysis and found numerous applications in various areas~\cite{Ne10}. The first step of finding central vertices is to define suitable indices measuring relative importance of vertices. Over the past decades, many centrality measures were introduced to characterize and analyze the roles of vertices in networks~\cite{WhSm03,BoVi14,BeKl15,BoDeRi16}. Among them, a popular one is closeness centrality~\cite{Ba48,Ba50}: the closeness of a vertex is the reciprocal of the sum of shortest path distances between it and all other vertices. However, this metric considers only the shortest paths, and more importantly, neglects contributions from other paths. Therefore it can produce some odd effects, or even counterintuitive results~\cite{BeWeLuMe16}. To avoid this shortcoming, Brandes and Fleischer presented current flow closeness centrality~\cite{BrFl05} based on electrical networks~\cite{DoSn84} , which takes into account contributions from all paths between vertices. Current flow based closeness has been shown to better discriminate vertices than its traditional counterparts~\cite{BeWeLuMe16}.

While most previous works focus on measures and algorithms for the importance of individual vertices in networks~\cite{WhSm03,LuChReZhZhZh16}, the problem of determining a group of $k$ most important vertices arises frequently in data mining and graph applications. For example, in social networks, retailers may want to choose $k$ vertices as promoters of product, such that the number of the potentially influenced customers is maximized~\cite{KeKlTa03}. Another example is P2P networks, where one wants to place resources on a fixed number of $k$ peers so they are easily accessed by others~\cite{GkMiSa06}. In order to measure the importance of a group of vertices, Everett and Borgatti~\cite{EvBo99} extended the idea of individual centrality to group centrality, and introduced the concepts of group centrality, for example, group closeness. Recently, some algorithms have been developed to compute or estimate group closeness~\cite{ZhLuToGu14,ChWaWa16, BeGoMe18}. However, similar to the case of individual vertices, these notions of group centrality also disregard contributions from paths that are not shortest.

In this paper, we extend current flow closeness of individual vertices~\cite{BeWeLuMe16} by proposing current flow closeness centrality (CFCC) for group of vertices. In a graph with $n$ vertices and $m$ edges, the CFCC $C(S)$ of a vertex group $S \subset V$ is equal to the ratio of $n$ to the sum of effective resistances between $S$ and all  vertices  $u$ in $V \setminus S$. 
We then consider the optimization problem: how can we find a group $S^*$ of $k$ vertices so as to maximize $C(S^*)$. We solve this problem by considering an equivalent problem of minimizing the reciprocal of $C(S^*)$. 
We show that the problem is NP-hard in Section~\ref{sec:nphard}, but also prove that the problem is an instance of supermodular set function optimization with cardinality constraint in Section~\ref{sec:supermodular}. The latter allows us to devise greedy algorithms to solve this problem, leading to two greedy algorithms with provable approximation guarantees:
\begin{enumerate}
\item A deterministic algorithm with a $(1-\frac{k}{k-1}\cdot\frac{1}{e})$  approximation factor and $O(n^3)$ running time (Section~\ref{sec:greedy});
\item  A randomized algorithm with a $(1-\frac{k}{k-1}\cdot\frac{1}{e}-\eps)$-approximation factor and $\Otil (km\eps^{-2})$ running time\footnote{We use the $\Otil (\cdot)$ notation to hide the ${\rm poly}\log $ factors.} for any small $\eps>0$
(Section~\ref{sec:algo}).
\end{enumerate}
A key ingredient of our second algorithm is nearly linear time solvers for Laplacians and symmetric, diagonally dominant, M-matrices (SDDM)~\cite{ST14,CKMPPRC14},
which has been used in various optimization problems on graphs~\cite{DaSp08,KeMiPe12,MiPe13}.

We perform extensive experiments on some networks to evaluate our algorithm, and some of their results are in Section~\ref{sec:experiments}. Our code is available on GitHub at \url{https://github.com/lchc/CFCC-maximization}. These results show that both algorithms are effective. Moreover, the second algorithm is efficient and is scalable to large networks with more than a million vertices.


\subsection{Related Works}

There exist various measures for centrality of a group of vertices, based on graph structure or dynamic processes, such as betweenness~\cite{DoelPuZi09, FiSp11,Yo14,MaTsUp16}, absorbing random-walk centrality~\cite{LiYuHuCh14,MaMagi15, ZhLiXiWuXuLu17}, and grounding centrality~\cite{PiSu14,ClHobuP017}. Since the criterion for importance of a vertex group is application dependent~\cite{GhTeLeYa14}, many previous works focus on selecting (or deleting) a group of $k$ vertices (for some given $k$) in order to optimize related quantities. These quantities are often measures of vertex group importance motivated by the applications, including minimizing the leading eigenvalue of adjacency matrix for vertex immunization~\cite{ToPrTsElFaCh10,ChToPrTsElFaCh16}, minimizing the mean steady-state variance for first-order leader-follower noisy consensus dynamics~\cite{PaBa10,ClPo11}, maximizing average distance for identifying structural hole spanners~\cite{ReLiXuLi5, XuReLiYuLi17}, and others.

Previous works on closeness centrality and related algorithms are most directly related to our focus on the group closeness centrality in this paper. The closeness centrality for an individual vertex was proposed~\cite{Ba48} and formalized~\cite{Ba50} by Bavelas. For a given vertex, its closeness centrality is defined as the reciprocal of the sum of shortest path distances of the vertex to all the other vertices. Everett and Borgatti~\cite{EvBo99} extended the individual closeness centrality to group closeness centrality, which measures how close a vertex group is to all other vertices. For a graph with $n$ vertices and $m$ edges, exactly computing the closeness centrality of a group of vertices involves calculating all-pairwise shortest path length, the time complexity of the state-of-the-art algorithm~\cite{Jo77} for which is $O(nm+n^2\log n)$.  To reduce the computation complexity, various approximation algorithms were developed.
A greedy $O(n^3)$ algorithm with approximation ratio $\left(1-\frac{k}{k-1}\cdot\frac{1}{e}\right)$ was devised~\cite{ChWaWa16}, and a sampling algorithm that scales better to large networks, but without approximation guarantee was also proposed in the same paper. Very recently, new techniques~\cite{BeGoMe18} have been developed to speed up the greedy algorithm in~\cite{ChWaWa16} while preserving its theoretical guarantees.

Conventional closeness centrality is based on the shortest paths, omitting the contributions from other paths. In order to overcome this drawback, Brandes and Fleischer introduced current flow closeness centrality for an individual vertex~\cite{BrFl05}, which essentially considers all paths between vertices, but still gives large weight to short paths.
Our investigation can be viewed as combining this line of current based centrality measures with the study of selecting groups of $k$ vertices.
For the former, a subset of the authors of this paper (Li and Zhang) recently demonstrated that current flow centrality measures for single edges can be computed provably efficiently~\cite{LZ18}.
Our approximation algorithm in Section~\ref{sec:algo} is directly motivated by that routine.

\section{Preliminaries}

In this section, we briefly introduce some  useful notations and tools
for the convenience of description of our problem and algorithms.

\subsection{Notations}

We use normal lowercase letters like $a,b,c$ to denote scalars in $\mathbb{R}$,
normal uppercase letters like $A, B, C$ to denote sets, bold lowercase letters like $\aa,\bb,\cc$ to denote vectors,
and bold uppercase letters like $\AA, \BB, \CC$ to denote matrices.
We write $\aa_{[i]}$ to denote the $i^{\mathrm{th}}$ entry of vector $\aa$
and $\AA_{[i,j]}$ to denote entry $(i,j)$ of matrix $\AA$.
We also write $\AA_{[i,:]}$ to denote the $i^{\mathrm{th}}$ row of $\AA$
and $\AA_{[:,j]}$ to denote the $j^{\mathrm{th}}$ column of $\AA$.

We write sets in matrix subscripts to denote submatrices. For example,
$\AA_{[I,J]}$ denotes the submatrix of $\AA$ with row indices in $I$
and column indices in $J$.
To simplify notation,
we also write $\AA_{-i}$ to denote the submatrix of $\AA$
obtained by removing the $i^{\mathrm{th}}$ row and $i^{\mathrm{th}}$ column of $\AA$.
For example, for an $n\times n$ matrix $\AA$,
$\AA_{-n}$ denotes the submatrix $\AA_{[1:n-1,1:n-1]}$.

Note that the precedence of matrix subscripts is
the lowest. Thus, $\AA_{-n}^{-1}$ denotes the inverse of $\AA_{-n}$
instead of a submatrix of $\AA^{-1}$.

For two matrices $\AA$ and $\BB$, we write $\AA \preceq \BB$ to denote
that $\BB - \AA$ is positive semidefinite, i.e.,
$\xx^T \AA \xx \leq \xx^T \BB \xx$ holds for every real vector $\xx$.

We use $\ee_i$ to denote the $i^{\mathrm{th}}$ standard basis vector of appropriate dimension,
and $\one_{S}$ to denote the indicator vector of $S$.

\subsection{Graphs, Laplacians, and Effective Resistances}

We write $G = (V,E,w)$ to denote a positively weighted undirected graph
with $n$ vertices, $m$ edges, and edge weight function $w : E \to \mathbb{R}^+$.
The Laplacian matrix $\LL$ of $G$ is defined as
$\LL_{[u,v]} = -w(u,v)$ if $u \sim v$,
$\LL_{[u,v]} = \mathrm{deg}(u)$ if $u = v$,
and $\LL_{[u,v]} = 0$ otherwise,
where $\mathrm{deg}(u) \defeq \sum\nolimits_{u\sim v} w(u,v)$
is the weighted degree of $u$ and $u\sim v$ means $(u,v) \in E$. Let $w_{\max}$ and  $w_{\min}$ denote, respectively, the maximum weight and minimum weight among all edges.
If we orient each edge of $G$ arbitrarily,
we can also write it's Laplacian as $\LL = \BB^T \WW \BB$,
where $\BB_{m\times n}$ is the signed edge-vertex incidence matrix defined by
$\BB_{[e,u]} = 1$ if $u$ is $e$'s head,
$\BB_{[e,u]} = -1$ if $u$ is $e$'s tail,
and $\BB_{[e,u]} = 0$ otherwise,
and $\WW_{m\times m}$ is a diagonal matrix with $\WW_{[e,e]} = w(e)$.
It is not hard to show that quadratic forms of $\LL$ can be written as
$
	\xx^T \LL \xx = \sum\nolimits_{u\sim v} w(u,v) \kh{\xx_{[u]} - \xx_{[v]}}^2,
$
which immediately implies that $\LL$ is positive semidefinite,
and $\LL$ only has one zero eigenvalue if $G$ is a connected graph.

The following fact shows that submatrices of Laplacians are always
positive definite and inverse-positive.

\begin{fact}\label{fact:inv+}
	Let $\LL$ be the Laplacian of a connected graph and
	let $\XX$ be a nonnegative, diagonal matrix with at least
	one nonzero entry. Then,
	$\LL + \XX$ is positive definite,
	and every entry of $\kh{\LL + \XX}^{-1}$ is positive.
\end{fact}

Let $0 = \lambda_1 < \lambda_2 \leq \ldots \leq\lambda_{n}$
be eigenvalues of $\LL$
of a connected graph $G$, and $\vv_1,\vv_2,\ldots,\vv_n$ be the corresponding
orthonormal eigenvectors. Then
we can  decompose $\LL$ as
$\LL = \sum\nolimits_{i=2}^n \lambda_i \vv_i \vv_i^T$
and define its pseudoinverse as
$\LL^\dag = \sum\nolimits_{i=2}^n \frac{1}{\lambda_i} \vv_i \vv_i^T$.

It is not hard to verify that if $\LL$ and $\HH$ are Laplacians of connected
graphs supported on the same vertex set,
then $\LL \preceq \HH$ implies
$\HH^\dag \preceq \LL^\dag$.

The pseudoinverse of  Laplacian matrix can be used to define effective resistance
between any pair of vertices~\cite{KlRa93}.

\begin{definition}
	For a connected graph $G = (V,E,w)$ with Laplacian matrix $\LL$,
	the effective resistance between vertices $u$ and $v$ is defined as
	$
		\er(u,v) = \kh{\ee_u - \ee_v}^T \LL^\dag \kh{\ee_u - \ee_v}\,.
	$
\end{definition}

The effective resistance between two vertices can also be expressed in term of the  diagonal elements
of the inverse for submatrices of $\LL$.

\begin{fact}[~\cite{IzKeWu13}]
	$
		\er(u,v) = \kh{\LL_{-u}^{-1}}_{[v,v]} = \kh{\LL_{-v}^{-1}}_{[u,u]}.
	$
\end{fact}

\subsection{Current Flow Closeness Centrality}

The current flow closeness centrality was proposed in~\cite{BrFl05}.
It is based on the assumption that information spreads efficiently
like an electrical current.

To define current flow closeness, we treat the graph $G$
as a resistor network via replacing
every edge $e$ by a resistor with resistance $r_e = 1 / w(e)$.
Let $v_{st}(u)$ denote the voltage of $u$ when
a unit current enters the network at $s$
and leaves it at $t$.

\begin{definition}
	\label{def:ccsingle}
	The current flow closeness $C(u)$ of a vertex $u$ is defined as
	$
		C(u) = n / \sum\nolimits_{v\in V}
			\kh{v_{uv}(u) - v_{uv}(v)}.
	$
\end{definition}

It has been proved~\cite{BrFl05} that the current flow closeness of vertex $u$
equals the ratio of $n$ to the
sum of effective resistances between $u$ and other vertices.

\begin{fact}\label{fact:cer}
	$C(u) = n / \sum\nolimits_{v\in V} \er(u,v)$.
\end{fact}

Actually, current flow closeness centrality is equivalent to information centrality~\cite{SZ89}.

\subsection{Supermodular Functions}

We now give the definitions for monotone and
supermodular set functions.
For simplicity, we write $S + u$ to denote $S\union\setof{u}$
and $S - u$ to denote $S\setminus \setof{u}$.

\begin{definition}[Monotonicity]
A set function $f : 2^V \to \mathbb{R}$ is
monotone if
$
	f(S)\geq f(T)
$
holds for all $S\subseteq T$.
\end{definition}

\begin{definition}[Supermodularity]
A set function $f : 2^V \to \mathbb{R}$ is
supermodular if
$
	f(S) - f(S + u) \geq f(T) - f(T + u)
$
holds for
all $S\subseteq T \subseteq V$ and $u\in V$.
\end{definition}

\section{Current Flow Closeness of a Group of  Vertices}\label{sec:groupcfc}

We follow the idea of~\cite{BrFl05} to define current flow closeness centrality (CFCC)
of a group of vertices.

To define current flow closeness centrality for a vertex set $S\subseteq V$,
we treat the graph $G$ as a resistor network in which
all vertices in $S$ are grounded. Thus, vertices in $S$
always have voltage $0$.
For a vertex $u\notin S$,
let $v_{uS}(v)$ be the voltage of $v$
when
a unit current enters the network at $u$
and leaves it at $S$ (i.e. the ground).
Then, we define the current flow closeness of $S$ as follows.

\begin{definition}
\label{CFCC} Let $G = (V,E,w)$ be a connected weighted graph.
 The current flow closeness centrality $C(S)$
of a vertex group $S\subseteq V$ is defined as
$
	C(S) = {n}/
		{\sum\nolimits_{u\in V} \kh{v_{uS}(u) - v_{uS}(S)}}
	= {n}/
		{\sum\nolimits_{u\in V} v_{uS}(u)}.
$
\end{definition}
Note that there are different variants of the definition of  CFCC for a vertex group. For example, we can use $ \kh{n-\mid S \mid} / \kh{\sum\nolimits_{u\in V} v_{uS}(u)}$ as the measure of CFCC for a vertex set $S$. Definition~\ref{CFCC} adopts the standard form as the classic closeness centrality~\cite{ChWaWa16}.

We next show that  $C(S)$ is in fact equal to the ratio of $n$ to a sum of effective resistances as in Fact~\ref{fact:cer}.

Let $u\in V$ be a fixed vertex.
Suppose there is a unit current enters the network at $u$
and leaves it at $S$.
Let $\vv \in \mathbb{R}^n$
be a vector of voltages at vertices.
By Kirchhoff's Current Law and Ohm's Law, we have
$
	\LL \vv = \ee_u - \ii_{S},
$
where $\ii_{S}$ denotes the amount of current flowing out of $S$.
Since vertices in $S$ all have voltage $0$,
we can restrict this equation to vertices in $V - S$ as
$
	\LL_{[V-S,V-S]} \vv_{[V-S]} = \ee_u 
$,
which leads to
$
	\vv_{[V-S]} = \LL_{[V-S,V-S]}^{-1} \ee_u.
$
This gives the expression of voltage at $u$ as
$
	v_{uS}(u) = \ee_u^T \vv = \ee_u^T \LL_{-S}^{-1} \ee_u.
$
Now we can write the  CFCC of $S$ as
\begin{align*}
	C(S) = \frac{n}
		{\sum\nolimits_{u\in V} v_{uS}(u)} =
	\frac{n}
		{\sum\nolimits_{u\in V - S} \ee_u^T \LL_{-S}^{-1} \ee_u}
	= \frac{n}{\trace{\LL_{-S}^{-1}}}.
\end{align*}

Note that the diagonal entry $\kh{\LL_{-S}^{-1}}_{[u,u]}$ of  $\kh{\LL_{-S}^{-1}}$  is exactly the effective resistance $\er(u,S)$ between vertex $u$ and vertex set $S$~\cite{ClPo11}, with $\er(u,S)=0$ for  any $u\in S$.
Then we have the following relation governing $C(S)$ and $\er(u,S)$.
\begin{fact}
	$C(S) = n / \trace{\LL_{-S}^{-1}}
	 = n / \sum\nolimits_{u\in V} \er(u,S)$.
\end{fact}

Being able to define CFCC of a vertex set
raises the problem of maximizing
current flow closeness subject to a cardinality constraint,
which we state below.

\begin{problem}[\underline{C}urrent \underline{F}low \underline{C}loseness \underline{M}aximization, CFCM]
	\label{prob:ccm}
	Given a connected graph $G = (V,E,w)$ with $n$ vertices, $m$ edges,
	and edge weight function $w : E \to \mathbb{R}^+$
	and an integer $1\leq k\leq n$, find a vertex group $S^* \in V$ such that the CFCC $C(S^*)$ is maximized, that is
		$
			S^* \in \underset{ {S\subseteq V,\sizeof{S} = k} }
				{\operatorname{arg \,max}} \quad
				C(S) \,.
	$
\end{problem}

\section{Hardness of Current Flow Closeness
	Maximization}\label{sec:nphard}
	
In this section, we prove that Problem~\ref{prob:ccm} is NP-hard.
We will give a reduction from vertex cover on 3-regular graphs
(graphs whose vertices all have degree 3),
which is an NP-complete problem~\cite{FHJ98}.
The decision version of this problem is stated below.
\begin{problem}[\underline{V}ertex \underline{C}over on \underline{3}-regular graphs,  VC3]
	Given a connected 3-regular graph $G = (V,E)$ and an integer $k$,
	decide whether or not there is a vertex set $S\subset V$ such
	that $\sizeof{S} \leq k$ and $S$ is a vertex cover of $G$
	(i.e. every edge in $E$ is incident with
	at least one vertex in $S$).
\end{problem}

An instance of this problem is denoted by VC3$(G,k)$.

We then give the decision version of Problem~\ref{prob:ccm}.

\begin{problem}[\underline{C}urrent \underline{F}low \underline{C}loseness \underline{M}aximization,
	\underline{D}ecision Version, CFCMD]
	Given a connected graph $G = (V,E,w)$, an integer $k$, and a real number $r \in \mathbb{R}$,
	decide whether or not there is a vertex set $S\subset V$ such
	that $\sizeof{S} \leq k$ and $C(S) \geq r$.
\end{problem}

An instance of this problem is denoted by CFCMD$(G,k,r)$.

To give the reduction,
we will need the following lemma.
	
\begin{lemma}\label{lem:vc3}
	Let $G = (V,E,w)$ be a connected 3-regular graph with all edge weights being $1$
	(i.e. $w(e) = 1$ for all $e\in E$).
	Let $S\subset V$ be a nonempty vertex set, and $k = \sizeof{S}$.
	Then, $C(S) \leq 3n / (n - k)$ and
	the equality holds if and only if $S$ is a vertex cover of $G$.
\end{lemma}

\begin{proof}
	We first show that if $S$ is a vertex cover of $G$ then $C(S) = 3n / (n - k)$.
	When $S$ is a vertex cover, $V\setminus S$ is an independent set.
	Thus, $\LL_{-S}$ is a diagonal matrix with all diagonal entries being $3$.
	So we have
	$
		C(S) = n / \trace{\LL_{-S}^{-1}}
		= n / \trace{\kh{\mathrm{diag}(3,\ldots,3)}^{-1}} = 3n / (n - k).
	$
	
	We then show that if $S$ is not a vertex cover of $G$ then $C(S) < 3n / (n - k)$.
	When $S$ is not a vertex cover, $V\setminus S$ is not an independent set.
	Thus, $\LL_{-S}$ is a block diagonal matrix,
	with each block corresponding to a connected component of $G[V\setminus S]$,
	the induced graph of $G$ on $V\setminus S$.
	Let $T\subseteq V\setminus S$ be a connected component of $G[V\setminus S]$
	such that $\sizeof{T} > 1$.
	Then, the block of $\LL_{-S}$ corresponding to $T$ is $\LL_{[T,T]}$.
	For a vertex $u\in T$,
	let the $u^{\textrm{th}}$ column of $\LL_{[T,T]}$ be
	$\begin{pmatrix} 3 \\ -\aa \end{pmatrix}$.
	Then, we can write $\LL_{[T,T]}$ into block form as
	$
		\LL_{[T,T]} = \begin{pmatrix}
			3 & -\aa^T \\
			-\aa & \AA
		\end{pmatrix},
	$
	where $\AA \defeq \LL_{[T - u,T - u]}$.
	By blockwise matrix inversion we have
	$
		\kh{\LL_{[T,T]}^{-1}}_{[u,u]} =
		1 / \kh{3 - \aa^T \AA^{-1}\aa}.
	$
	Since $\LL_{[T,T]}^{-1}$ is positive definite, we have
	$1 / (3 - \aa^T \AA^{-1}\aa) > 0$ and hence $\aa^T \AA^{-1}\aa < 3$.
	Since $T$ is a connected component, $\aa$ is not a zero vector,
	which coupled with the fact that $\AA^{-1}$ is positive definite
	gives $\aa^T \AA^{-1}\aa > 0$. Thus,
	$1 / (3 - \aa^T \AA^{-1}\aa) > 1 / 3$.
	Since this holds for all $u\in T$, we have $\trace{\LL_{[T,T]}^{-1}} > \sizeof{T}/3$.
	Also, since $T$ can be any connected component of $G[V\setminus S]$ with at least two vertices,
	and a block of an isolate vertex in $G[V\setminus S]$ contributes a $1/3$ to $\trace{\kh{\LL_{-S}}^{-1}}$,
	we have
	for any $S$ which is not a vertex cover of $G$
	$
		\trace{\LL_{-S}^{-1}} > \sizeof{V\setminus S} / 3 = (n - k) / 3,
	$
	which implies $C(S) = n / \trace{\LL_{-S}^{-1}} < 3n / (n - k)$.
\end{proof}

The following theorem then follows by Lemma~\ref{lem:vc3}.


\begin{theorem}
Maximizing current flow closeness subject to a cardinality constraint is NP-hard.
\end{theorem}

\begin{proof}
	We give a polynomial reduction
	$
		p : \setof{(G = (V,E), k)} \to \setof{(G = (V,E,w), k, r)}
	$
	from instances of VC3 to instances of CFCMD.
	For a connected 3-regular graph $G = (V,E)$ with $n$ vertices,
	we construct a weighted graph $G' = (V,E,w_1)$
	with the same vertex set and edge set and
	an edge weight function $w_1 : E \to \setof{1}$ mapping all edges to weight $1$.
	Then, we construct a reduction $p$
	as
	\begin{align*}
		p\kh{(G = (V,E), k)} = \kh{G' = (V,E,w_1), k, 3n / (n - k)}.
	\end{align*}
	By Lemma~\ref{lem:vc3}, $p$ is a polynomial reduction from VC3 to CFCMD,
	which implies that CFCM is NP-hard.
\end{proof}

\section{Supermodularity of the Reciprocal of Current Flow Group Closeness}
\label{sec:supermodular}

In this section, we prove that the
reciprocal of current flow group closeness,
i.e., $\trace{\LL_{-S}^{-1}} / n$, is a monotone supermodular function.
Our proof uses the following lemma,
which shows that $\LL_{-S}^{-1}$ is entrywise supermodular.

\begin{lemma}\label{lem:ewsup}
	Let $u,v \in V$ be an arbitrary pair of vertices.
	Then,
	the entry $\kh{\LL_{-S}^{-1}}_{[u,v]}$
	is a monotone supermodular function.
	Namely, for
	vertices $u,v \neq w$
	and nonempty vertex sets $S\subseteq T \subseteq V$
	such that $u,v,w \notin T$,
	$\kh{\LL_{-S}^{-1}}_{[u,v]} \geq \kh{\LL_{-T}^{-1}}_{[u,v]}$
	and
	\[
		\kh{\LL_{-S}^{-1}}_{[u,v]} - \kh{\LL_{-(S+w)}^{-1}}_{[u,v]}
		\geq \kh{\LL_{-T}^{-1}}_{[u,v]} - \kh{\LL_{-(T+w)}^{-1}}_{[u,v]}.
	\]
\end{lemma}

\ifx\aaa\undefined
To prove Lemma~\ref{lem:ewsup},
we first define a linear relaxation
$\HH : [0,1]^n \to \mathbb{R}^{n\times n}$
of
$\LL_{-S}$
as
\begin{align}\label{eq:defh}
	\kh{\HH(\xx)}_{[u,v]} =
	\begin{cases}
		\LL_{[u,u]} & \text{if $u = v$,} \\
		\kh{1 - \xx_{[u]}}\cdot \kh{1 - \xx_{[v]}}\cdot \LL_{[u,v]}
		& \text{if $u\neq v$.}
	\end{cases}
\end{align}
\begin{remark}
	We remark the intuition behind this relaxation $\HH(\xx)$.
	Let $\xx_{[=1]}$ denote the indices
	of entries of $\xx$ equal to one,
	and let
	$\xx_{[<1]}$ denote the indices of entries of $\xx$
	less than one.
	Then, by the definition in~(\ref{eq:defh}),
	we can write $\HH(\xx)$ into a block diagonal matrix as
	\begin{align*}
		\HH(\xx) =
		\begin{pmatrix}
			\kh{\HH(\xx)}_{[\xx_{[=1]},\xx_{[=1]}]} & \zero \\
			\zero &
			\kh{\HH(\xx)}_{[\xx_{[<1]},\xx_{[<1]}]}
		\end{pmatrix},
	\end{align*}
	where
	$\kh{\HH(\xx)}_{[\xx_{[=1]},\xx_{[=1]}]}$
	is itself a diagonal matrix.
	This means that
	if $\xx = \one_{S}$ for some
	nonempty vertex set $S\subseteq V$,
	the following statement holds:
	\[
		\kh{\HH^{-1}(\one_{S})}_{-S}
		=
		\LL_{-S}^{-1}.
	\]
	The condition that every entry of $\xx$ is in $[0,1]$
	coupled with Fact~\ref{fact:inv+}
	also implies that
	all submatrices of $\HH(\xx)$ are
	positive definite and inverse-positive.
\end{remark}
Now for
vertices $u,v \neq w$
and nonempty vertex set $S\subseteq V$
such that $u,v,w \notin S$,
we can write the marginal gain of a vertex $w$ as
\begin{align}\label{eq:marginal}
	\kh{\LL_{-S}^{-1} -
	\LL_{-(S+w)}^{-1}}_{[u,v]}
	=
	\kh{\HH^{-1}(\one_{S}) - \HH^{-1}(\one_{S+w})}_{[u,v]}.
\end{align}
We can further write the matrix on
the rhs of~(\ref{eq:marginal}) as an integral by
\begin{align}
	& \HH^{-1}(\one_{S}) - \HH^{-1}(\one_{S+w}) =
	- \left. \HH^{-1}(\one_{S} + t \cdot \ee_{w}) \right|_{0}^{1}
	\notag \\
	=
	& -\int_{0}^{1}
	\frac{\mathrm{d} \HH^{-1}(\one_{S} + t \cdot \ee_{w})}
	{\mathrm{d} t}
	\mathrm{d} t \notag \\
	=
	&
	\int_{0}^{1}
	\HH^{-1}(\one_{S} + t \cdot \ee_{w})
	\frac{\mathrm{d} \HH(\one_{S} + t \cdot \ee_{w})}
	{\mathrm{d} t}
	\HH^{-1}(\one_{S} + t \cdot \ee_{w}) \mathrm{d} t \notag \\
	=
	& \int_{0}^{1}
	\HH^{-1}(\one_{S} + t \cdot \ee_{w})
	\kh{ -t\HH(\one_V - \ee_{w}) }
	\HH^{-1}(\one_{S} + t \cdot \ee_{w}) \label{eq:int}
	\mathrm{d} t,
\end{align}
where the second equality follows by the identity
\[
	\frac{\mathrm{d} \AA^{-1}}{\mathrm{d} x}
	=
	- \AA^{-1}
	\frac{\mathrm{d} \AA}{\mathrm{d} x}
	\AA^{-1}
\]
for any invertible matrix $\AA$.

To prove Lemma~\ref{lem:ewsup},
we will also need the following lemma, which
shows the entrywise monoticity of $\HH^{-1}(\xx)$.

\begin{lemma}\label{lem:monoh}
	For $0\leq t\leq 1$,
	the following statement
	holds for any vertices $u,v\neq w$ and
	nonempty vertex sets $S\subseteq T\subseteq V$ such that
	$u,v,w\notin T$:
	\begin{align*}
		\kh{\HH^{-1}(\one_{S} + t \cdot \ee_{w})}_{[u,v]}
		\geq
		\kh{\HH^{-1}(\one_{T} + t \cdot \ee_{w})}_{[u,v]}.
	\end{align*}
\end{lemma}

\begin{proof}
For simplicity,  we let
	$\SS \defeq \HH(\one_{S} + t \cdot \ee_{w})$ and
	$\TT \defeq \HH(\one_{T} + t \cdot \ee_{w})$.
	We also write $P \defeq V\setminus T$,
	$Q \defeq V\setminus S$,
	and $F \defeq T\setminus S$.
	Due to the block diagonal structures of $\SS$ and $\TT$,
	we have
	\begin{align}\label{eq:diags}
		\kh{\SS^{-1}}_{[u,v]} =
		\kh{\SS_{[Q,Q]}^{-1}}_{[u,v]}
	\end{align}
	and
	\begin{align}\label{eq:diagt}
		\kh{\TT^{-1}}_{[u,v]} =
		\kh{\TT_{[P,P]}^{-1}}_{[u,v]}.
	\end{align}
	Since $\SS$ and $\TT$ agree
	on entries with indices in $P$,
	we can write the submatrix $\SS_{[Q,Q]}$ of $\SS$
	in block form as
	\begin{align*}
		\SS_{[Q,Q]} =
		\begin{pmatrix}
			\TT_{[P,P]} & \SS_{[P,F]} \\
			\SS_{[F,P]} &
			\SS_{[F,F]}
		\end{pmatrix}.
	\end{align*}
	By blockwise matrix inversion, we have
	\begin{align*}
		& \kh{\SS_{[Q,Q]}^{-1}}_{[P,P]} \\
		= & \TT_{[P,P]}^{-1} +
		\TT_{[P,P]}^{-1}
		\SS_{[P,F]}
		\kh{\SS_{[Q,Q]}^{-1}}_{[F,F]}
		\SS_{[F,P]}
		\TT_{[P,P]}^{-1} \\
		= &
		\TT_{[P,P]}^{-1} +
		\TT_{[P,P]}^{-1}
		\kh{-\SS_{[P,F]}}
		\kh{\SS_{[Q,Q]}^{-1}}_{[F,F]}
		\kh{-\SS_{[F,P]}}
		\TT_{[P,P]}^{-1},
	\end{align*}
	where the second equality follows by
	negating both $\SS_{[P,F]}$ and $\SS_{[F,P]}$.
	By definition
	the matrix $-\SS_{[P,F]}$ is entrywise nonnegative.
	By Fact~\ref{fact:inv+},
	every entry of
	$\SS_{[Q,Q]}^{-1}$
	and $\TT_{[P,P]}^{-1}$ is also nonnegative.
	Thus, the matrix
	\[
		\TT_{[P,P]}^{-1}
		\kh{-\SS_{[P,F]}}
		\kh{\SS_{[Q,Q]}^{-1}}_{[F,F]}
		\kh{-\SS_{[F,P]}}
		\TT_{[P,P]}^{-1}
	\]
	is entrywise nonnegative,
	which coupled with~(\ref{eq:diags}) and~(\ref{eq:diagt}),
	implies $\kh{\SS^{-1}}_{[u,v]} \geq 	\kh{\TT^{-1}}_{[u,v]}$.
\end{proof}

\begin{proof}[Proof of Lemma~\ref{lem:ewsup}]

By definition the matrix $-t\HH(\one_V - \ee_{w})$ is entrywise nonnegative
when $t\geq 0$.
By Fact~\ref{fact:inv+},
the matrix $\HH^{-1}(\one_{S} + t \cdot \ee_{w})$ is also entrywise nonnegative
when $0\leq t\leq 1$.
Thus, the derivative in~(\ref{eq:int}) is entrywise nonnegative,
which implies the the monotonicity of $\kh{\LL_{-S}^{-1}}_{[u,v]}$
for any pair of vertices $u,v \in V$.

We then prove the supermodularity, i.e.,
\begin{align}\label{eq:ewsup}
	\kh{\LL_{-S}^{-1} -
	\LL_{-(S+w)}^{-1}}_{[u,v]}
	\geq
	\kh{\LL_{-T}^{-1} -
	\LL_{-(T+w)}^{-1}}_{[u,v]}
\end{align}
for any $S\subseteq T\subseteq V$
and $u,v\notin T$.
Lemma~\ref{lem:monoh},
coupled with the fact that $\HH^{-1}(\one_{S} + t \cdot \ee_{w})$
and $-t\HH(\one_V - \ee_{w})$
are both entrywise nonnegative,
gives the entrywise monotonicity of the derivative
in~(\ref{eq:int}) as
\begin{align*}
	& \kh{\HH^{-1}(\one_{S} + t \cdot \ee_{w})
	\kh{ -t\HH(\one_V - \ee_{w}) }
	\HH^{-1}(\one_{S} + t \cdot \ee_{w})}_{[u,v]} \\
	\geq \
	& \kh{\HH^{-1}(\one_{T} + t \cdot \ee_{w})
	\kh{ -t\HH(\one_V - \ee_{w}) }
	\HH^{-1}(\one_{T} + t \cdot \ee_{w})}_{[u,v]}.
\end{align*}
Integrating both sides of
the above inequality with respect to $t$ on the interval $[0,1]$
gives~(\ref{eq:ewsup}).
\end{proof}
\fi

The following theorem follows by Lemma~\ref{lem:ewsup}.

\begin{theorem}\label{thm:tracesup}
The reciprocal of current flow group centrality,
	i.e., $\trace{\LL_{-S}^{-1}} / n$,
	is a monotone supermodular function.
\end{theorem}
\begin{proof}
	Let $S\subseteq T\subseteq V$ be vertex sets
	and $w\notin T$ be a vertex.
	
	For monotonicity, we have
	\begin{align*}
		\trace{\LL_{-S}^{-1}} &=
		\sum\limits_{u\notin S} \kh{\LL_{-S}^{-1}}_{[u,u]}
		\geq
		\sum\limits_{u\notin T} \kh{\LL_{-S}^{-1}}_{[u,u]}\\
		& \geq
		\sum\limits_{u\notin T} \kh{\LL_{-T}^{-1}}_{[u,u]}
		=
		\trace{\LL_{-T}^{-1}},
	\end{align*}
	where the first inequality follows by the fact
	that $\LL_{-S}^{-1}$ is entrywise nonnegative,
	and the second inequality follows from
	the entrywise monotinicity of $\LL_{-S}^{-1}$.
	
	For supermodularity,
	we have
	\begin{align*}
		& \trace{\LL_{-S}^{-1}} -
		\trace{\LL_{-(S+w)}^{-1}} \\
		= & \kh{\LL_{-S}^{-1}}_{[w,w]} +
			\sum\limits_{u\notin (S+w)} \kh{\LL_{-S}^{-1}
			- \LL_{-(S+w)}^{-1}}_{[u,u]} \\
		\geq & \kh{\LL_{-T}^{-1}}_{[w,w]} +
			\sum\limits_{u\notin (T+w)} \kh{\LL_{-S}^{-1}
			- \LL_{-(S+w)}^{-1}}_{[u,u]} \\
		\geq & \kh{\LL_{-T}^{-1}}_{[w,w]} +
			\sum\limits_{u\notin (T+w)} \kh{\LL_{-T}^{-1}
			- \LL_{-(T+w)}^{-1}}_{[u,u]} \\
		= & \trace{\LL_{-T}^{-1}} -
		\trace{\LL_{-(T+w)}^{-1}},
	\end{align*}
	where the first inequality follows from
	the entrywise monotonicity of $\LL_{-S}^{-1}$,
	and the second inequality follows from
	the entrywise supermodularity of $\LL_{-S}^{-1}$.
\end{proof}

We note that~\cite{ClPo11} has previously proved that $\trace{\LL_{-S}^{-1}}$
is monotone and supermodular by using the connection between effective resistance and commute time for random walks.
However, our proof is fully algebraic. Moreover, we  present
a more general result  that $\LL_{-S}^{-1}$ is entrywise supermodular.

Theorem~\ref{thm:tracesup} indicates that one can obtain a $(1 - \frac{k}{k-1}\cdot\frac{1}{e})$-approximation
to the optimum $\trace{\LL_{-S^*}^{-1}}$
by a simple greedy algorithm, by
picking the vertex with the maximum marginal gain each time~\cite{NeWoFi78}.  However, since computing $\trace{\LL_{-S}^{-1}}$ involves matrix inversions, a naive implementation of this greedy algorithm
will take $O(k n^4 )$ time, assuming that one matrix inversion runs in $O(n^3)$ time.
We will show in the next section how to implement this greedy algorithm
in $O(n^3)$ time using blockwise matrix inversion.

\section{A Deterministic Greedy Algorithm}\label{sec:greedy}

We now consider how to accelerate the naive greedy algorithm.
Suppose that after the $i^{\mathrm{th}}$ step,
the algorithm has selected a set $S_i$
containing $i$ vertices.  We next   compute the marginal gain
$\trace{\LL_{-S_i}^{-1}}	-	\trace{\LL_{-(S_i+u)}^{-1}}$
of each vertex $u\notin S_i$.

For a vertex $u\notin S_i$, let
$\begin{pmatrix} d_u \\ -\aa \end{pmatrix}$
denote the $u^{\mathrm{th}}$ column of the submatrix
$\LL_{-S_i}$. Then we write $\LL_{-S_i}$
in block form as
$
	\LL_{-S_i} =
	\begin{pmatrix}
		d_u & -\aa^T \\
		-\aa & \AA
	\end{pmatrix},
$
where $\AA \defeq \LL_{-(S_i+u)}$.
By blockwise matrix inversion, we have
\begin{align}\label{eq:block}
	\LL_{-S_i}^{-1} =
	\begin{pmatrix}
		\frac{1}{s} &
		\frac{1}{s} \aa^T \AA^{-1} \\
		\frac{1}{s} \AA^{-1} \aa &
		\AA^{-1} +
		\frac{1}{s}
		\AA^{-1}
		\aa \aa^T
		\AA^{-1}
	\end{pmatrix},
\end{align}
where $s = d_u - \aa^T \AA^{-1} \aa$.  Then the marginal gain of $u$  can be further expressed as
\begin{align*}
	&\trace{\LL_{-S_i}^{-1}} - \trace{\LL_{-(S_i+u)}^{-1}}
	\\
	= & \kh{\LL_{-S_i}^{-1}}_{[u,u]}
		+ \sum\limits_{v\in V\setminus (S_i + u)}
			\kh{ \kh{\LL_{-S_i}^{-1}}_{[v,v]} -
			\kh{\LL_{-(S_i + u)}^{-1}}_{[v,v]}} \\
	= & \frac{1}{s} + \frac{1}{s}
		\trace{\AA^{-1} \aa \aa^T \AA^{-1}}
	= \frac{1}{s} + \frac{1}{s} \aa^T \AA^{-1} \AA^{-1} \aa \\
	= &
		\kh{\ee_u^T \LL_{-S_i}^{-2} \ee_u}/
		\kh{\ee_u^T \LL_{-S_i}^{-1} \ee_u},
\end{align*}
where the second equality and the fourth equality
follow by~(\ref{eq:block}),
while the third equality follows by the cyclicity of trace.

By~(\ref{eq:block}), we can  also update the inverse
$\LL_{-S_i}^{-1}$ upon a vertex $u$ by
\begin{align*}
	\LL_{-(S_i + u)}^{-1} =
	\kh{
	\LL_{-S_i}^{-1} -
	\kh{ \LL_{-S_i}^{-1} \ee_u \ee_u^T \LL_{-S_i}^{-1}} /
	\kh{\ee_u^T \LL_{-S_i}^{-1} \ee_u}}_{-u}.
\end{align*}

At the first step,
we need to pick a vertex $u_1$ with minimum 
$\sum\nolimits_{v\in V}\er(u_1,v)$,
which can be done by computing $\sum\nolimits_{v\in V}\er(u,v)$
for all $u\in V$ using the relation~\cite{BoFr13}
$
	\sum\nolimits_{v\in V}\er(u,v) =
	n \kh{\LL^\dag}_{[u,u]} +
	\trace{\LL^\dag}.
$

We give the $O(n^3 + k n^2)$-time algorithm as follows.
\begin{algbox}
	$S_k = \exactGreedy\kh{G,\LL,k}$
	
	\quad
	\begin{enumerate}
		\item Compute $\LL^\dag$ by inverting $\LL$ in $O(n^3)$ time.
		\item $S_1 \gets \setof{u_1}$ where
			$
				u_1 = \argmin_{u\in V}\nolimits
					n\kh{\LL^\dag}_{[u,u]} + \trace{\LL^\dag}.
			$
		\item Compute $\LL_{-S_1}^{-1}$ in $O(n^3)$ time.
		\item Repeat the following steps for $i = 1,\ldots, k-1$:
			\begin{enumerate}
				\item
					$
						u_{i+1} \gets \argmax_{u\in (V\setminus S_i)}\nolimits
							\frac{\ee_u^T \LL_{-S_i}^{-2} \ee_u}
							{\ee_u^T \LL_{-S_i}^{-1} \ee_u}
					$, $S_{i+1} \gets S_i + u_{i+1}$
				\item Compute $\LL_{S_{i+1}}^{-1}$ in $O(n^2)$ time by
					\[
						\LL_{-(S_i + u_{i+1})}^{-1} =
							\kh{
							\LL_{-S_i}^{-1} -
							\frac{\LL_{-S_i}^{-1} \ee_{u_{i+1}}
							\ee_{u_{i+1}}^T \LL_{-S_i}^{-1}}
							{\ee_{u_{i+1}}^T \LL_{-S_i}^{-1} \ee_{u_{i+1}}}}_
							{-u_{i+1}}.
					\]
			\end{enumerate}
		\item Return $S_k$.
	\end{enumerate}
\end{algbox}

The performance of $\exactGreedy$ is characterized
in the following theorem.

\begin{theorem}
	The algorithm $S_k = \exactGreedy\kh{G,\LL,k}$
	takes an undirected positive weighted graph $G = (V,E,w)$
	with associated Laplacian $\LL$ and an integer $2\leq k\leq n$,
	and returns a vertex set $S_k\subseteq V$ with $\sizeof{S_k} = k$.
	The algorithm runs in time $O(n^3)$.
	The vertex set $S_k$ satisfies
\begin{align*}
		\trace{\LL_{-u^*}^{-1}} -
		\trace{\LL_{-S_k}^{-1}}
		\geq
		\kh{1 - \frac{k}{k-1}\cdot \frac{1}{e}}
		\kh{\trace{\LL_{-u^*}^{-1}} -
		\trace{\LL_{-S^*}^{-1}}
		}
\end{align*}
	where
	$S^* \defeq
		\argmin\limits_{\sizeof{S}\leq k}
		\trace{\LL_{-S}^{-1}}$
	and
	$u^* \defeq \argmin\limits_{u\in V}
		\sum\limits_{v\in V}\er(u,v)$.
\end{theorem}
\begin{proof}
	The running time is easy to verify.
	We only need to prove the approximation ratio.
	
By supermodularity, for any $i\geq 1$
	\begin{align*}
		\trace{\LL_{-S_i}^{-1}} -
		\trace{\LL_{-S_{i+1}}^{-1}}
		\geq \frac{1}{k}
		\kh{\trace{\LL_{-S_i}^{-1}} -
			\trace{\LL_{-S*}^{-1}}},
	\end{align*}
	which implies
	\begin{align*}
		\trace{\LL_{-S_{i+1}}^{-1}} -
			\trace{\LL_{-S^*}^{-1}} \leq
			\kh{1 - \frac{1}{k}}
			\kh{ \trace{\LL_{-S_{i}}^{-1}} -
			\trace{\LL_{-S^*}^{-1}}}.
	\end{align*}
	Then, we have
	\begin{align*}
		\trace{\LL_{-S_{k}}^{-1}} - \trace{\LL_{-S^*}^{-1}} &\leq
		\kh{1 - \frac{1}{k}}^{k-1}
		\kh{ \trace{\LL_{-S_{1}}^{-1}} - \trace{\LL_{-S^*}^{-1}}} \\
		&\leq
		\frac{k}{k-1} \cdot \frac{1}{e}
		\kh{ \trace{\LL_{-S_{1}}^{-1}} - \trace{\LL_{-S^*}^{-1}}},
	\end{align*}
	which coupled with
	$\trace{\LL_{-S_{1}}^{-1}} =
		\trace{\LL_{-u^*}^{-1}}$ completes the proof. 
\end{proof}

\section{A Randomized Greedy Algorithm}\label{sec:algo}

The deterministic greedy algorithm  $\exactGreedy$ has a time complexity $O(n^3)$, which is   still  not acceptable  for large networks. In this section, we provide an efficient randomized  algorithm, which achieves a  $(1 - \frac{k}{k-1}\cdot\frac{1}{e}-\eps)$  approximation factor in time $\Otil(km\eps^{-2})$.

 To further accelerate algorithm $\exactGreedy$ we need to
compute the marginal gains
\begin{align}\label{eq:stepi}
	\trace{\LL_{-S_i}^{-1}}
	-
	\trace{\LL_{-(S_i+u)}^{-1}} =
	\frac{\ee_u^T \LL_{-S_i}^{-2} \ee_u}
							{\ee_u^T \LL_{-S_i}^{-1} \ee_u}
\end{align}
for all $u\in V$
and a vertex set $S_i\subseteq V$ more quickly.
We also need a faster way to compute
\begin{align}\label{eq:step1}
	\sum\nolimits_{v\in V}\er(u,v) =
	n \kh{\LL^\dag}_{[u,u]} +
	\trace{\LL^\dag}
\end{align}
for all $u\in V$ at the $1^{\mathrm{st}}$ step.
We will show how to solve both problems in nearly linear
time using Johnson-Lindenstrauss Lemma
and Fast
SDDM Solvers.
Our routines are motivated by the effective resistance estimation routine
in~\cite{SS11,KoLePe16}.

\begin{lemma}[Johnson-Lindenstrauss Lemma~\cite{JL84}]
	\label{lem:jl}
	\quad Let \\$\vv_1,\vv_2,\cdots,\vv_n \in \mathbb{R}^d$
	be fixed vectors
	and $0 < \eps < 1$ be a real number.
	Let $q$ be a positive integer such that
	$
		q \geq 4(\eps^2 / 2 - \eps^3 / 3)^{-1} \ln n
	$
	and $\QQ_{q \times d}$ be a random matrix obtained by
	first choosing each of its entries from
	Gaussian distribution $N(0,1)$
	independently
	and then
	normalizing each of
	its columns to a length of $1$.
	With high probability, the following statement holds
	for any $1\leq i, j\leq n$:
	\[
			(1 - \eps) \norm{\vv_i - \vv_j}^2
			\leq \norm{\QQ \vv_i - \QQ \vv_j}^2 \leq
			(1 + \eps) \norm{\vv_i - \vv_j}^2.
	\]
\end{lemma}

\begin{lemma}[Fast SDDM Solvers~\cite{ST14,CKMPPRC14}]\label{lem:solve}
	There is a routine $\xx = \SDDMSolver(\SS, \bb, \eps)$
	which takes
	a Laplacian or an SDDM matrix
	$\SS_{n\times n}$
	with $m$ nonzero entries,
	a vector $\bb \in \mathbb{R}^n$, and
	an error parameter $\delta > 0$,
	and returns a vector $\xx \in \mathbb{R}^n$ such that
	$
		\norm{\xx - \SS^{-1} \bb}_{\SS}
		\leq
		\delta \norm{\SS^{-1} \bb}_{\SS}
	$
	holds with high probability,
	where $\norm{\xx}_{\SS} \defeq \sqrt{\xx^T \SS \xx}$,
	and $\SS^{-1}$ denotes the pseudoinverse of $\SS$
	when $\SS$ is a Laplacian.
	The routine runs in expected time
	$\Otil(m\log(1/\delta))$.
\end{lemma}

\subsection{Approximation of~(\ref{eq:step1})}



To approximate~(\ref{eq:step1}) we
need to approximate all diagonal entries of $\LL^\dag$,
i.e., $\ee_u^T \LL^\dag \ee_u$ for all $u\in V$.
We first write $\ee_u^T \LL^\dag \ee_u$ in an euclidian norm
as
\begin{align*}
	&\ee_u^T \LL^\dag \ee_u =
	\ee_u^T \LL^\dag \LL \LL^\dag \ee_u
	=
	\ee_u^T \LL^\dag \BB^T \WW \BB \LL^\dag \ee_u \\
 	= &\ee_u^T \LL^\dag \BB^T \WW^{1/2}
 		\WW^{1/2} \BB \LL^\dag \ee_u
	= \norm{\WW^{1/2} \BB \LL^\dag \ee_u}^2.
\end{align*}

Then we use Johnson-Lindenstrauss Lemma to reduce the dimensions.
Let $\QQ_{q\times m}$ be a random Gaussian matrix
where $q = \ceil{4(\eps^2 / 2 - \eps^3 / 3)^{-1} \ln n}$.
By Lemma~\ref{lem:jl},
\begin{align*}
	\norm{\WW^{1/2} \BB \LL^\dag \ee_u}^2 \approx_{1+\eps}
	\norm{\QQ \WW^{1/2} \BB \LL^\dag \ee_u}^2
\end{align*}
holds for all $u\in V$ with high probability.
Here we can use sparse matrix multiplication to compute
$\QQ \WW^{1/2} \BB$, and Fast SDDM solvers to compute
$\QQ \WW^{1/2} \BB \LL^\dag$.

We give the routine for approximating~(\ref{eq:step1})
as follows.

\begin{algbox}
	$\left\{ r_u \right\}_{u\in V} =
		\ERSumEst\kh{G,\LL, \eps}$
		
	\quad
	\begin{enumerate}
		\item Set
			$
				\delta = \frac{\eps}{9n^2}
						\kh{ \frac{ (1 - \eps/3) w_{\mathrm{min}} }
						{ (1 + \eps/3) w_{\mathrm{max}} } }^{1/2}
			$.
		\item Generate a random Gaussian matrix
			$\QQ_{q\times m}$ where
				$q = \ceil{4((\eps/3)^2 / 2 - (\eps/3)^3 / 3)^{-1} \ln n}$.
		\item Compute $\kh{\QQ \WW^{1/2} \BB}_{q\times n}$
			by sparse matrix multiplication in $O(qm)$ time.
		\item Compute an approximation $\ZZtil_{q\times n}$ to
			$\ZZ_{q\times n} \defeq \QQ \WW^{1/2} \BB \LL^\dag$ by
			\[
				\ZZtil_{[i,:]}
				\gets
				\SDDMSolver(
					\LL, \kh{\QQ \WW^{1/2} \BB}_{[i,:]}^T,
						\delta)^T
			\]
			where $1\leq i\leq q$.
		\item $r_u \gets \norm{\ZZtil}_F^2 + \norm{\ZZtil \ee_u}^2$
			for all $u\in V$ and return $\left\{ r_u \right\}_{u\in V}$.
	\end{enumerate}
\end{algbox}

\begin{lemma}\label{lem:diaglb}
	Let $\LL$ be the Laplacian of a connected graph.
	Then for any $u\in V$,
	\[
		\ee_u^T \LL^\dag \ee_u \geq \frac{n-1}{n^2 w_{\mathrm{max}}}.
	\]
\end{lemma}

\begin{proof}
$\ee_u^T \LL^\dag \ee_u$ can be evaluated as
	\begin{align*}
		&\quad \ee_u^T \LL^\dag \ee_u= \ee_u^T \PPi \LL^\dag \PPi \ee_u
		= \kh{\ee_u - \frac{1}{n}\one}^T \LL^\dag
			\kh{\ee_u - \frac{1}{n}\one} \\
		&\geq \frac{1}{\lambda_n \kh{\LL}}
			\norm{\ee_u - \frac{1}{n}\one}^2
		\geq \frac{1}{\lambda_n \kh{w_{\mathrm{max}}\cdot \LL^{K_n}}}
			\norm{\ee_u - \frac{1}{n}\one}^2 \\
		& = \frac{1}{n w_{\mathrm{max}}} \cdot \frac{n-1}{n}
		  = \frac{n-1}{n^2 w_{\mathrm{max}}},
	\end{align*}
	where the second inequality follows by
	$\LL \preceq w_{\mathrm{max}} \cdot \LL^{K_n}$.
\end{proof}


\begin{lemma}\label{lem:ersumsest}
	The routine $\ERSumEst$
   runs in time $\Otil(m)$.
	For $0 < \eps \leq 1/2$,
	the $\left\{ r_u \right\}_{u\in V}$ returned by
	$\ERSumEst$ satisfies
	\[
		(1-\eps) \sum\nolimits_{v\in V}\er(u,v) \leq
		r_u \leq
		(1+\eps) \sum\nolimits_{v\in V}\er(u,v)
	\]
	with high probability.
\end{lemma}

\begin{proof}
	It suffices to prove that
	\begin{align}\label{eq:zapprox}
		(1 - \eps/3) \norm{\ZZ \ee_u}^2 \leq
		\norm{\ZZtil \ee_u}^2
		\leq
		(1 + \eps/3) \norm{\ZZ \ee_u}^2
	\end{align}
	holds for all $u\in V$,
	where $\ZZ$ and $\ZZtil$ are defined in the routine $\ERSumEst$.
	Let $P_{uv}$ denote an arbitrary simple path connecting vertices
	$u$ and $v$.
	We fist upper bound the difference between
	the square roots of the above two values,
	i.e., $\norm{\ZZ \ee_u}$ and $\norm{\ZZtil \ee_u}$:
	\begin{align*}
		&
		\abs{\norm{\ZZ \ee_u} - \norm{\ZZtil \ee_u}}
		\leq \norm{(\ZZ - \ZZtil) \ee_u}
		\qquad \text{by the triangle inequality} \\
		= & 
		\norm{(\ZZ - \ZZtil) \kh{\ee_u - \frac{1}{n}\one}}
		= n^{-1/2}
		\norm{(\ZZ - \ZZtil)
			\sum\nolimits_{v\neq u} \kh{\ee_u - \ee_v}} \\
		\leq &
		n^{-1/2} \sum\limits_{v\neq u}
		\norm{(\ZZ - \ZZtil) \kh{\ee_u - \ee_v}} \\
		\leq &
		n^{-1/2} \sum\limits_{v\neq u}
		\sum\limits_{(a,b) \in P_{uv}}
		\norm{(\ZZ - \ZZtil) \kh{\ee_a - \ee_b}} 
		\qquad \text{by applying the triangle inequality twice} \\
		\leq & \kh{n \sum\limits_{v\neq u}
		\sum\limits_{(a,b) \in P_{uv}}
		\norm{(\ZZ - \ZZtil) \kh{\ee_a - \ee_b}}^2 }^{1/2}
		\qquad  \text{by Cauchy-Schwarz} \\
		\leq & n \kh{\sum\limits_{(a,b)\in E}
			\norm{(\ZZ - \ZZtil) \kh{\ee_a - \ee_b}}^2}^{1/2} \\
		= & n \norm{(\ZZ - \ZZtil) \BB^T }_F \qquad
			\text{$\norm{\cdot}_F$ denotes Frobenius norm} \\
		\leq &
		 n w_{\mathrm{min}}^{-1/2}
			\norm{(\ZZ - \ZZtil) \BB^T \WW^{1/2}}_F
			\qquad \text{since $\norm{\WW} \geq w_{\mathrm{min}}$} \\
		\leq &
		n w_{\mathrm{min}}^{-1/2} \delta
			\norm{\ZZ \BB^T \WW^{1/2}}_F
			\qquad \text{by Lemma~\ref{lem:solve}} \\
		= & n w_{\mathrm{min}}^{-1/2} \delta
		\kh{
		  \sum\limits_{(a,b) \in E} w(a,b) \cdot
		  	\norm{\ZZ (\ee_u - \ee_v)}}^{1/2}
		  \\
		\leq & n w_{\mathrm{min}}^{-1/2} \delta
		\kh{
		  (1 + \eps/3)
		  \sum\limits_{(a,b) \in E} w(a,b) \er(a,b) }^{1/2}
		  	\qquad \text{by Lemma~\ref{lem:jl}} \\
		= &
		n \delta \kh{\frac{(1 + \eps/3)(n-1)}{w_{\mathrm{min}}}}^{1/2}.
	\end{align*}
	By Lemma~\ref{lem:jl} and Lemma~\ref{lem:diaglb}, we have
	\begin{align*}
		\norm{\ZZ \ee_u}^2 \geq
		(1 - \eps/3) \ee_u^T \LL^\dag \ee_u \geq
		(1 - \eps/3) \frac{n - 1}{n^2 w_{\mathrm{max}}},
	\end{align*}
	which combined with the upper bound above and
	the fact that
	\[
		\delta = \frac{\eps}{9n^2}
			\kh{ \frac{ (1 - \eps/3) w_{\mathrm{min}} }
				{ (1 + \eps/3) w_{\mathrm{max}} } }^{1/2}
	\] gives$\abs{\norm{\ZZ \ee_u} - \norm{\ZZtil \ee_u}}
		\leq (\eps/9) \norm{\ZZ \ee_u}$.	Then, we have
	\begin{align*}
		&\abs{\norm{\ZZ \ee_u}^2 - \norm{\ZZtil \ee_u}^2}
		= \abs{\norm{\ZZ \ee_u} - \norm{\ZZtil \ee_u}}
			\kh{\norm{\ZZ \ee_u} + \norm{\ZZtil \ee_u}} \\
		\leq &
		(\eps/9) \norm{\ZZ \ee_u} \cdot
		(2 + (\eps/9)) \norm{\ZZ \ee_u}
		\leq
		(\eps/3) \norm{\ZZ \ee_u}^2,
	\end{align*}
	which implies~(\ref{eq:zapprox}).
\end{proof}

\subsection{Approximation of~(\ref{eq:stepi})}

We first approximate  the numerator of~(\ref{eq:stepi}), which can be recast
in an euclidian norm as $	\ee_u^T \LL_{-S_i}^{-2} \ee_u =	\norm{\LL_{-S_i}^{-1} \ee_u}^2$.
We then once again
use Johnson-Lindenstrauss Lemma to reduce the dimensions.
Let
$\PP_{p\times n}$ be a random Gaussian matrix
where $p = \ceil{4(\eps^2 / 2 - \eps^3 / 3)^{-1} \ln n}$.
By Lemma~\ref{lem:jl}, we have that for all $u\in V$,
\begin{align*}
	(1 - \eps)
	\norm{\LL_{-S_i}^{-1} \ee_u}^2 \leq
	\norm{\PP \LL_{-S_i}^{-1} \ee_u}^2 \leq
	(1 + \eps)
	\norm{\LL_{-S_i}^{-1} \ee_u}^2,
\end{align*}
where $\PP \LL_{-S_i}^{-1}$ can be computed using Fast SDDM Solvers.

We continue to approximate the denominator of~(\ref{eq:stepi}).
Since $\LL_{-S_i}$ is an SDDM matrix,
we can express it in terms of the  sum of a Laplacian and a nonnegative
diagonal matrix as $\LL_{-S_i} = \BB'^T \WW' \BB' + \XX$.
Then we can write
\begin{align}
	&\ee_u^T \LL_{-S_i}^{-1} \ee_u
	= \ee_u^T \LL_{-S_i}^{-1} \LL_{-S_i} \LL_{-S_i}^{-1} \ee_u \notag \\
	= &\ee_u^T \LL_{-S_i}^{-1} \kh{\BB'^T \WW' \BB' + \XX}
		\LL_{-S_i}^{-1} \ee_u \notag \\
	= &
	\ee_u^T \LL_{-S_i}^{-1} \BB'^T \WW' \BB'
		\LL_{-S_i}^{-1} \ee_u +
	\ee_u^T \LL_{-S_i}^{-1} \XX
		\LL_{-S_i}^{-1} \ee_u \notag \\
	= &
	\norm{\WW'^{1/2} \BB' \LL_{-S_i}^{-1} \ee_u}^2 +
	\norm{\XX^{1/2} \LL_{-S_i}^{-1} \ee_u}^2. \label{eq:deno}
\end{align}
Let
$\QQ_{q\times m}$ and $\RR_{r\times n}$
be random Gaussian matrices
where $q = r = \ceil{4(\eps^2 / 2 - \eps^3 / 3)^{-1} \ln n}$.
By Lemma 7.1, we have for all $u\in V$
\begin{align*}
	(\ref{eq:deno}) \approx_{1 + \eps}
	\norm{\QQ \WW'^{1/2} \BB' \LL_{-S_i}^{-1} \ee_u}_2^2 +
	\norm{\RR \XX^{1/2} \LL_{-S_i}^{-1} \ee_u}_2^2.
\end{align*}
Here $\QQ \WW'^{1/2} \BB' \LL_{-S_i}^{-1}$ and
$\RR \XX^{1/2} \LL_{-S_i}^{-1}$ can be computed
using Fast SDDM Solvers.

We give the routine for approximating~(\ref{eq:stepi}) as follows.
\begin{algbox}
	$\left\{ g_u \right\}_{u\in V} =
		\GainsEst\kh{G,\LL, S_i, \eps}$
		
	\quad
	\begin{enumerate}
		\item Set $\delta_1 = \frac{\wmin \eps}{27 \wmax n^4}
			\kh{\frac{ \kh{1 - \eps/9} }{ (1+\eps/9)n } }^{1/2}$ and
			$\delta_2 = \delta_3 = \frac{1}{4n^4 \wmax}
				\kh{\frac{\eps \wmin^{3/2}}{9n}}^{1/2}$.
		\item Generate random Gaussian matrices
			$\PP_{p\times n}, \QQ_{q\times m}, \RR_{r\times n}$
			where $p,q,r =
				\ceil{4((\eps/9)^2 / 2 - (\eps/9)^3 / 3)^{-1} \ln n}$.
		\item Let $\XX = \mathrm{Diag}\kh{\LL_{-S_i} \one}$
			and denote $\LL_{-S_i} - \XX$ by $\BB'^T \WW' \BB'$.
		\item Compute $\QQ \WW'^{1/2} \BB'$ and $\RR \XX^{1/2}$
			by sparse matrix multiplication in $O(pm)$ time.
		\item Compute approximations
			$\ZZtil^{(1)}$ to
				$\ZZ^{(1)} \defeq \PP \LL_{-S_i}^{-1}$,
			$\ZZtil^{(2)}$ to
				$\ZZ^{(2)} \defeq \QQ \WW'^{1/2} \BB' \LL_{-S_i}^{-1}$, and
			$\ZZtil^{(3)}$ to
				$\ZZ^{(3)} \defeq \RR \XX^{1/2} \LL_{-S_i}^{-1}$
			by
			\begin{enumerate}
				\item $\ZZtil^{(1)}_{[i,:]}
					\gets \SDDMSolver(\LL_{-S_i}, \PP_{[i,:]}^T, \delta_1)^T$,
				\item $\ZZtil^{(2)}_{[i,:]}
					\gets \SDDMSolver(\LL_{-S_i},
					\kh{\QQ \WW'^{1/2} \BB'}_{[i,:]}^T, \delta_2)^T$, and
				\item $\ZZtil^{(3)}_{[i,:]}
					\gets \SDDMSolver(\LL_{-S_i},
					\kh{\RR \XX^{1/2}}_{[i,:]}^T, \delta_3)^T$.
			\end{enumerate}
		\item $g_u \gets
			\frac{\norm{\ZZtil^{(1)}\ee_u}^2}
			{\norm{\ZZtil^{(2)}\ee_u}^2 + \norm{\ZZtil^{(3)}\ee_u}^2}$
			for all $u\in V$ and return $\left\{ g_u \right\}_{u\in V}$.
	\end{enumerate}
\end{algbox}

\begin{lemma}\label{lem:error2}
	Let $\SS_{n\times n}$ be a positive definite matrix,
	$\MM_{d\times n}$ be an arbitrary matrix, and
	$\QQ_{q \times d}$ be a random matrix obtained by
	first choosing each of its entries from
	Gaussian distribution $N(0,1)$
	independently
	and then
	normalizing each of
	its columns to a length of $1$
	where
	\[
		q \geq 4(\eps^2 / 2 - \eps^3 / 3)^{-1} \ln n.
	\]
	Let $\ZZ_{q\times n} = \QQ \MM \SS^{-1}$.
	If $\ZZtil_{q\times n}$ is a matrix
	such that for all $1\leq i\leq q$
	\begin{align}\label{eq:deltaZ}
		\norm{\ZZ^T \ee_i - \ZZtil^T \ee_i}_{\SS}^2
		\leq
		\delta \norm{\ZZ^T \ee_i}_{\SS}^2,
	\end{align}
	then
	\[
		\norm{\ZZ - \ZZtil}_F^2 \leq
		(1 + \eps)
		\frac{\delta^2}{\lambdamin(\SS)}
		\trace{\SS^{-1}} \trace{\MM^T \MM}
	\]
	with high probability.
\end{lemma}

\begin{proof}
	We first upper bound the Frobenius norm of
	$\ZZ - \ZZtil$ by
	\begin{align*}
		& \norm{\ZZ - \ZZtil}_F^2 =
			\sum\limits_{i=1}^q
			\norm{\ZZ^T\ee_i - \ZZtil^T\ee_i}^2 \\
		\leq &
			\frac{1}{\lambdamin(\SS)}
			\sum\limits_{i=1}^q
			\norm{\ZZ^T\ee_i - \ZZtil^T\ee_i}_{\SS}^2
			\qquad
			\text{since $\lambdamin(\SS)\cdot\II \pleq \SS$}\\
		\leq &
			\frac{\delta^2}{\lambdamin(\SS)}
			\sum\limits_{i=1}^q
			\norm{\ZZ^T\ee_i}_{\SS}^2
			\qquad
			\text{by~(\ref{eq:deltaZ}).}
	\end{align*}
	We then upper bound
	$\sum\limits_{i=1}^q	\norm{\ZZ^T\ee_i}_{\SS}^2$ by
	\begin{align*}
		& \sum\limits_{i=1}^q	\norm{\ZZ^T\ee_i}_{\SS}^2 =
		\sum\limits_{i=1}^q \ee_i^T \QQ \MM \SS^{-1} \MM^T \QQ^T \ee_i
		=
		\trace{\QQ \MM \SS^{-1} \MM^T \QQ^T} \\
		= &
		\norm{\QQ \kh{\MM \SS^{-1} \MM^T}^{1/2}}_F^2
		\leq
		(1 + \eps) \norm{\kh{\MM \SS^{-1} \MM^T}^{1/2}}_F^2
		\qquad \text{by Lemma~\ref{lem:jl}} \\
		= &
		(1 + \eps) \trace{\MM \SS^{-1} \MM^T}
		= (1 + \eps) \trace{\SS^{-1} \MM^T \MM}
		\qquad \text{by cyclicity of trace} \\
		\leq &
		(1 + \eps) \trace{\SS^{-1}} \trace{\MM^T \MM}
		\qquad  \text{by submultiplicativity of PSD's trace.}
	\end{align*}
	Combining the above two upper bounds
	completes the proof.
\end{proof}

\begin{lemma}\label{lem:traces}
	For any nonempty $S\subseteq V$,
	$\trace{\LL_{-S}^{-1}} \leq n^2/\wmin$.
\end{lemma}
\begin{proof}
	We upper bound $\trace{\LL_{-S}^{-1}}$ by
	\begin{align*}
		\trace{\LL_{-S}^{-1}}
		= \sum\limits_{u\in V-S} \er(u,S)
		\leq \sum\limits_{i=1}^{n-1} i\cdot \frac{1}{\wmin}
		\leq \frac{n^2}{\wmin}.
	\end{align*}
\end{proof}

\begin{lemma}\label{lem:lambdas}
	For any nonempty $S\subseteq V$,
	$\lambda_{\mathrm{min}}\kh{\LL_{-S}} \geq w_{\mathrm{min}} / n^2$
	and
	$\lambda_{\mathrm{max}}\kh{\LL_{-S}} \leq n^2 w_{\mathrm{max}} $.
\end{lemma}
\begin{proof}
	We upper bound $\lambdamax$ using Cauchy interlacing:
	\begin{align*}
		\lambdamax\kh{\LL_{-S}} \leq
		\lambdamax\kh{\LL} \leq
		\lambdamax\kh{\LL^{K_n}}
		\leq
		n^2 w_{\mathrm{max}}.
	\end{align*}
	The lower bound of $\lambdamin$ follows by
	\begin{align*}
		\lambdamin\kh{\LL_{-S}}
		\geq 1 / \trace{\LL_{-S}^{-1}} \geq w_{\mathrm{min}} / n^2.
	\end{align*}
\end{proof}

\begin{lemma}\label{lem:gainsest}
	The routine $\GainsEst$ runs in time $\Otil(m)$.
	For $0\leq \eps \leq 1/2$,
	the $\left\{ g_u \right\}_{u\in V}$ returned by
	$\GainsEst$ satisfies
	\[
		(1-\eps) \frac{\ee_u^T \LL_{-S_i}^{-2} \ee_u}
							{\ee_u^T \LL_{-S_i}^{-1} \ee_u} \leq
		g_u \leq
		(1+\eps) \frac{\ee_u^T \LL_{-S_i}^{-2} \ee_u}
							{\ee_u^T \LL_{-S_i}^{-1} \ee_u}
	\]
	with high probability.
\end{lemma}


\begin{proof}
	It suffices to show that
	\begin{align}
		(1 - \eps/9) \norm{\ZZ^{(1)} \ee_u}^2 \leq
		\norm{\ZZtil^{(1)} \ee_u}^2
		&\leq
		(1 + \eps/9) \norm{\ZZ^{(1)} \ee_u}^2 \label{eq:z1}
	\end{align}
	and
	\begin{align}\label{eq:z3}
	   \norm{\ZZ^{(2)} \ee_u}^2 + \norm{\ZZ^{(3)} \ee_u}^2
	   \approx_{(1 + \eps/9)}
		\norm{\ZZtil^{(2)} \ee_u}^2 + \norm{\ZZtil^{(3)} \ee_u}^2
	\end{align}
	hold for all $u\in V$.
	
	By Lemma~\ref{lem:error2},~\ref{lem:traces}, and~\ref{lem:lambdas},
	and $\trace{\II} = n$,
	$\trace{\BB'^T \WW' \BB'} \leq \trace{\LL} \leq n^2 \wmax$,
	and $\trace{\XX} \leq n \wmax$, we have
	\begin{align}
		\norm{\ZZ^{(1)} - \ZZtil^{(1)}}_F^2 &\leq
		(1 + \eps/9)
		\frac{\delta_1^2 n^5}{\wmin^2}, \label{eq:dfz1} \\
		\norm{\ZZ^{(2)} - \ZZtil^{(2)}}_F^2 &\leq
		(1 + \eps/9)
		\frac{\delta_2^2 n^6 \wmax}{\wmin^2}, \label{eq:dfz2}\\
		\norm{\ZZ^{(3)} - \ZZtil^{(3)}}_F^2 &\leq
		(1 + \eps/9)
		\frac{\delta_3^2 n^5 \wmax}{\wmin^2}. \label{eq:dfz3}
	\end{align}
	By Lemma~\ref{lem:lambdas} and Lemma~\ref{lem:jl} we have
	\begin{align}
		\norm{\ZZ^{(1)} \ee_u}^2 &\geq
		(1 - \eps/9) \ee_u^T \LL_{-S_i}^{-2} \ee_u
		\geq
		\frac{1 - \eps/9}{n^4 \wmax^2}, \\
		\norm{\ZZ^{(2)} \ee_u}^2 + \norm{\ZZ^{(3)} \ee_u}^2
		&\geq
		(1 - \eps/9) \ee_u^T \LL_{-S_i}^{-1} \ee_u
		\geq \frac{1 - \eps/9}{n^2 \wmax} \label{eq:lbz23}.
	\end{align}
	We upper bound $\abs{\norm{\ZZ^{(1)}\ee_u} - \norm{\ZZtil^{(1)}\ee_u}}$
	by
	\begin{align*}
		\abs{\norm{\ZZ^{(1)}\ee_u} - \norm{\ZZtil^{(1)}\ee_u}}
		\leq &\norm{ \kh{ \ZZ^{(1)} - \ZZtil^{(1)} } \ee_u }
		\qquad \text{by the triangle inequality} \\
		\leq &
		\norm{\ZZ^{(1)} - \ZZtil^{(1)}}_F
		\leq
		\sqrt{(1 + \eps/9)
		\frac{\delta_1^2 n^5}{\wmin^2}}
		\leq
		(\eps / 27) \norm{\ZZ^{(1)} \ee_u},
	\end{align*}
	which gives
	\begin{align*}
		& \abs{\norm{\ZZ^{(1)}\ee_u }^2 - \norm{\ZZtil^{(1)}\ee_u}^2} \\
		= &\abs{\norm{\ZZ^{(1)}\ee_u } - \norm{\ZZtil^{(1)}\ee_u}}
			\kh{\norm{\ZZ^{(1)}\ee_u } + \norm{\ZZtil^{(1)}\ee_u}} \\
		\leq &(\eps / 27) \norm{\ZZ^{(1)}\ee_u }
			(2 + \eps / 27) \norm{\ZZ^{(1)}\ee_u }
		\leq (1 + \eps/9) \norm{\ZZ^{(1)}\ee_u }^2.
	\end{align*}
	We then upper bound
	the difference between
	$\norm{\ZZ^{(2)}\ee_u }^2
	+ \norm{\ZZ^{(3)}\ee_u }^2$
	and
	$\norm{\ZZtil^{(2)}\ee_u}^2 + \norm{\ZZtil^{(3)}\ee_u}^2$ by
	\begin{align}
		& \abs{ \norm{\ZZ^{(2)}\ee_u }^2 + \norm{\ZZ^{(3)}\ee_u }^2 -
		\norm{\ZZtil^{(2)}\ee_u}^2 - \norm{\ZZtil^{(3)}\ee_u}^2} \\
		\leq &
		\abs{ \norm{\ZZ^{(2)}\ee_u }^2 - \norm{\ZZtil^{(2)}\ee_u}^2 } +
		\abs{ \norm{\ZZ^{(3)}\ee_u }^2 - \norm{\ZZtil^{(3)}\ee_u}^2 } \notag\\
		= &
			\abs{\norm{\ZZ^{(2)}\ee_u } - \norm{\ZZtil^{(2)}\ee_u}}
			\kh{\norm{\ZZ^{(2)}\ee_u } + \norm{\ZZtil^{(2)}\ee_u}} + \notag\\
			& \abs{\norm{\ZZ^{(3)}\ee_u } - \norm{\ZZtil^{(3)}\ee_u}}
			\kh{\norm{\ZZ^{(3)}\ee_u } + \norm{\ZZtil^{(3)}\ee_u}} \notag\\
		\leq &
		2\max\left\{
				\abs{\norm{\ZZ^{(2)}\ee_u } - \norm{\ZZtil^{(2)}\ee_u}},
				\abs{\norm{\ZZ^{(3)}\ee_u } - \norm{\ZZtil^{(3)}\ee_u}}
			\right\}^2 + \notag\\
		& 2\max\left\{
				\abs{\norm{\ZZ^{(2)}\ee_u } - \norm{\ZZtil^{(2)}\ee_u}},
				\abs{\norm{\ZZ^{(3)}\ee_u } - \norm{\ZZtil^{(3)}\ee_u}}
			\right\}\cdot
			\kh{ \norm{\ZZ^{(2)}\ee_u } + \norm{\ZZ^{(3)}\ee_u } }.
				\label{eq:qq}
	\end{align}
	Using Cauchy-Schwarz and
	Lemma~\ref{lem:jl},
	we can upper bound $\norm{\ZZ^{(2)}\ee_u } + \norm{\ZZ^{(3)}\ee_u }$
	by
	\begin{align*}
		&\norm{\ZZ^{(2)}\ee_u } + \norm{\ZZ^{(3)}\ee_u }
		\leq \sqrt{2\kh{\norm{\ZZ^{(2)}\ee_u }^2 + \norm{\ZZ^{(3)}\ee_u}^2}}\\
		\leq &
		\sqrt{2(1 + \eps/9) \ee_u^T \LL_{-S_i} \ee_u}
		\leq
		\sqrt{2(1 + \eps/9) \frac{n^2}{\wmin}}.
	\end{align*}
	Combining this upper bound
	with~(\ref{eq:dfz2}),~(\ref{eq:dfz3}),~(\ref{eq:lbz23})
	and the values of $\delta_2$ and $\delta_3$
	gives
	\begin{align*}
		& \abs{ \norm{\ZZ^{(2)}\ee_u }^2 + \norm{\ZZ^{(3)}\ee_u }^2 -
		\norm{\ZZtil^{(2)}\ee_u}^2 - \norm{\ZZtil^{(3)}\ee_u}^2}
		\leq
		(\eps/9)
		\kh{ \norm{\ZZ^{(2)}\ee_u }^2 + \norm{\ZZ^{(3)}\ee_u }^2 },
	\end{align*}
	which completes the proof.
\end{proof}

We now give the $\Otil(mk)$-time greedy algorithm
as follows.
\begin{algbox}
	$S_k = \FastGreedy\kh{G,\LL,k, \eps}$
	
	\quad
	\begin{enumerate}
		\item $\left\{ r_u \right\}_{u\in V} \gets
			\ERSumEst\kh{G,\LL, \eps/3}$
		\item $S_1 \gets \setof{u_1}$ where
			$
				u_1 = \argmin_{u\in V}\nolimits
					r_u.
			$
		\item Repeat the following steps for $i = 1,\ldots, k-1$:
			\begin{enumerate}
				\item $\left\{ g_u \right\}_{u\in V} \gets
					\GainsEst\kh{G,\LL, S_i, \eps/2}$
				\item $S_{i+1} \gets S_i + u_{i+1}$ where
					$
						u_{i+1} = \argmax_{u\in (V\setminus S_i)}\nolimits
							g_u.
					$
			\end{enumerate}
		\item Return $S_k$.
	\end{enumerate}
\end{algbox}

The performance of $\FastGreedy$ is characterized
in the following theorem.

\begin{theorem}
	The algorithm $S_k = \FastGreedy\kh{G,\LL,k,\eps}$
	takes an undirected positively weighted graph $G = (V,E,w)$
	with associated Laplacian $\LL$, an integer $2\leq k\leq n$,
	and an error parameter $0 < \eps \leq 1/2$,
	and returns a vertex set $S_k\subseteq V$ with $\sizeof{S_k} = k$.
	The algorithm runs in time $\Otil(k m)$.
	With high probability, the vertex set $S_k$ satisfies
	\begin{align}\label{eq:1/e+eps}
		(1 + \eps) \trace{\LL_{-u^*}^{-1}} -
		\trace{\LL_{-S_k}^{-1}}
		\geq
		\kh{1 - \frac{k}{k-1}\cdot \frac{1}{e} - \eps}
		\kh{(1 + \eps)\trace{\LL_{-u^*}^{-1}} -
		\trace{\LL_{-S^*}^{-1}}
		},
	\end{align}
	where
	$S^* \defeq
		\argmin\limits_{\sizeof{S}\leq k}
		\trace{\LL_{-S}^{-1}}$
	and
	$u^* \defeq \argmin\limits_{u\in V}
		\sum\limits_{v\in V}\er(u,v)$.
\end{theorem}
\begin{proof}
	
	The running time is easy to verify.
	
	The running time follows by the running times in
	Lemma~\ref{lem:ersumsest} and~\ref{lem:gainsest}.
	
	We now prove the approximation ratio.
	The main difference between $\FastGreedy$ and $\exactGreedy$
	is that at each step,\\
	$\exactGreedy$ picks a vertex with maximum
	marginal gain, while $\FastGreedy$ picks a vertex
	with at least $\frac{1 - \eps/2}{1 + \eps/2} \geq 1 - \eps$
	times maximum marginal gain.
	By supermodularity we have for any $i\geq 1$
	\begin{align*}
		\trace{\LL_{-S_i}^{-1}} -
		\trace{\LL_{-S_{i+1}}^{-1}}
		\geq \frac{1 - \eps}{k}
		\kh{\trace{\LL_{-S_i}^{-1}} -
			\trace{\LL_{-S*}^{-1}}},
	\end{align*}
	which implies
	\begin{small}
	\begin{align*}
		\trace{\LL_{-S_{i+1}}^{-1}} -
			\trace{\LL_{-S^*}^{-1}} \leq
			\kh{1 - \frac{1 - \eps}{k}}
			\kh{ \trace{\LL_{-S_{i}}^{-1}} -
			\trace{\LL_{-S^*}^{-1}}}.
	\end{align*}
	\end{small}
	Then, we have
	\begin{align*}
	&\quad	\trace{\LL_{-S_{k}}^{-1}} - \trace{\LL_{-S^*}^{-1}}\\ &\leq
		\kh{1 - \frac{1 - \eps}{k}}^{k-1}
		\kh{ \trace{\LL_{-S_{1}}^{-1}} - \trace{\LL_{-S^*}^{-1}}} \\
		&\leq
		\frac{k}{k-1} \cdot \frac{1}{e^{(1-\eps)}}
		\kh{ \trace{\LL_{-S_{1}}^{-1}} - \trace{\LL_{-S^*}^{-1}}} \\
		&\leq
		\frac{k}{k-1} \cdot \kh{\frac{1}{e} + \eps}
		\kh{ \trace{\LL_{-S_{1}}^{-1}} - \trace{\LL_{-S^*}^{-1}}}.
	\end{align*}
	By Lemma~\ref{lem:ersumsest}, we obtain
	\[
		\trace{\LL_{-S_{1}}^{-1}} \leq
		 \frac{1 + \eps/3}{1 - \eps/3}
		 \trace{\LL_{-u^*}^{-1}}
		 \leq
		 \kh{1 + \eps}\trace{\LL_{-u^*}^{-1}},
	\]
	which together with the above inequality
	implies~(\ref{eq:1/e+eps}).
\end{proof}

\section{Experiments}\label{sec:experiments}

\begin{table}[H]
\setlength{\abovecaptionskip}{0.pt}
\caption{Information of datasets. For a network with $n$ vertices and $m$ edges, we use
$n'$ and $m'$ to denote the number of vertices and
edges in its largest connected component, respectively. }\label{SetNo}
\begin{center}
\normalsize
\resizebox{0.6\columnwidth}{!}{
\begin{tabular}{ccccc}
\Xhline{2.5\arrayrulewidth}
\raisebox{-0.5ex}{Network} & \raisebox{-0.5ex}{$n$} & \raisebox{-0.5ex}{$m$} & \raisebox{-0.5ex}{$n'$} & \raisebox{-0.5ex}{$m'$} \\[0.5ex]
\hline
\raisebox{-0.5ex}{Zachary karate club} & \raisebox{-0.5ex}{34} & \raisebox{-0.5ex}{78} & \raisebox{-0.5ex}{34} & \raisebox{-0.5ex}{78} \\[0.5ex]
\raisebox{-0.5ex}{Windsufers} & \raisebox{-0.5ex}{43} & \raisebox{-0.5ex}{336} & \raisebox{-0.5ex}{43} & \raisebox{-0.5ex}{336} \\[0.5ex]

\raisebox{-0.5ex}{Contiguous USA} & \raisebox{-0.5ex}{49} & \raisebox{-0.5ex}{107} & \raisebox{-0.5ex}{49} & \raisebox{-0.5ex}{107} \\[0.5ex]

\raisebox{-0.5ex}{Barab\'asi-Albert} & \raisebox{-0.5ex}{50} & \raisebox{-0.5ex}{94} & \raisebox{-0.5ex}{50} & \raisebox{-0.5ex}{94} \\[0.5ex]
\raisebox{-0.5ex}{Watts-Strogatz} & \raisebox{-0.5ex}{50} & \raisebox{-0.5ex}{100} & \raisebox{-0.5ex}{50} & \raisebox{-0.5ex}{100} \\[0.5ex]

\raisebox{-0.5ex}{Erd{\"o}s-R{\'e}nyi} & \raisebox{-0.5ex}{50} & \raisebox{-0.5ex}{95} & \raisebox{-0.5ex}{50} & \raisebox{-0.5ex}{95} \\[0.5ex]

\raisebox{-0.5ex}{Regular ring lattice} & \raisebox{-0.5ex}{50} & \raisebox{-0.5ex}{100} & \raisebox{-0.5ex}{50} & \raisebox{-0.5ex}{100} \\[0.5ex]

\raisebox{-0.5ex}{Dolphins} & \raisebox{-0.5ex}{62} & \raisebox{-0.5ex}{159} & \raisebox{-0.5ex}{62} & \raisebox{-0.5ex}{159} \\[0.5ex]

\raisebox{-0.5ex}{David Copperfield} & \raisebox{-0.5ex}{112} & \raisebox{-0.5ex}{425} & \raisebox{-0.5ex}{112} & \raisebox{-0.5ex}{425} \\[0.5ex]
\raisebox{-0.5ex}{Jazz musicians} & \raisebox{-0.5ex}{198} & \raisebox{-0.5ex}{2742} & \raisebox{-0.5ex}{195} & \raisebox{-0.5ex}{1814} \\[0.5ex]
\raisebox{-0.5ex}{Virgili} & \raisebox{-0.5ex}{1,133} & \raisebox{-0.5ex}{5,451} & \raisebox{-0.5ex}{1,133} & \raisebox{-0.5ex}{5,451} \\[0.5ex]
\raisebox{-0.5ex}{Euroroad} & \raisebox{-0.5ex}{1,174} & \raisebox{-0.5ex}{1,417} & \raisebox{-0.5ex}{1,039} & \raisebox{-0.5ex}{1,305} \\[0.5ex]

\raisebox{-0.5ex}{Protein} & \raisebox{-0.5ex}{1,870} & \raisebox{-0.5ex}{2,277} & \raisebox{-0.5ex}{1,458} & \raisebox{-0.5ex}{1,948} \\[0.5ex]
\raisebox{-0.5ex}{Hamster full} & \raisebox{-0.5ex}{2,426} & \raisebox{-0.5ex}{16,631} & \raisebox{-0.5ex}{2,000} & \raisebox{-0.5ex}{16,098} \\[0.5ex]
\raisebox{-0.5ex}{ego-Facebook} & \raisebox{-0.5ex}{2,888} & \raisebox{-0.5ex}{2,981} & \raisebox{-0.5ex}{2,888} & \raisebox{-0.5ex}{2,981} \\[0.5ex]

\raisebox{-0.5ex}{Vidal} & \raisebox{-0.5ex}{3,133} & \raisebox{-0.5ex}{6,726} & \raisebox{-0.5ex}{2,783} & \raisebox{-0.5ex}{6,007} \\[0.5ex]
\raisebox{-0.5ex}{Powergrid} & \raisebox{-0.5ex}{4,941} & \raisebox{-0.5ex}{6,594} & \raisebox{-0.5ex}{4,941} & \raisebox{-0.5ex}{6,594} \\[0.5ex]

\raisebox{-0.5ex}{Reactome} & \raisebox{-0.5ex}{6,327} & \raisebox{-0.5ex}{147,547} & \raisebox{-0.5ex}{5,973} & \raisebox{-0.5ex}{145,778} \\[0.5ex]
%
%

\raisebox{-0.5ex}{ca-HepTh} & \raisebox{-0.5ex}{9,877} & \raisebox{-0.5ex}{25,998} & \raisebox{-0.5ex}{8,638} & \raisebox{-0.5ex}{24,806} \\[0.5ex]

\raisebox{-0.5ex}{PG-Privacy} & \raisebox{-0.5ex}{10,680} & \raisebox{-0.5ex}{24,316} & \raisebox{-0.5ex}{10,680} & \raisebox{-0.5ex}{24,316} \\[0.5ex]

\raisebox{-0.5ex}{CAIDA} & \raisebox{-0.5ex}{26,475} & \raisebox{-0.5ex}{53,381} & \raisebox{-0.5ex}{26,475} & \raisebox{-0.5ex}{53,381} \\[0.5ex]
\raisebox{-0.5ex}{ego-Twitter} & \raisebox{-0.5ex}{81,306} & \raisebox{-0.5ex}{1,342,296} & \raisebox{-0.5ex}{81,306} & \raisebox{-0.5ex}{1,342,296} \\[0.5ex]

\raisebox{-0.5ex}{com-DBLP} & \raisebox{-0.5ex}{317,080} & \raisebox{-0.5ex}{1,049,866} & \raisebox{-0.5ex}{317,080} & \raisebox{-0.5ex}{1,049,866} \\[0.5ex]

\raisebox{-0.5ex}{roadNet-PA} & \raisebox{-0.5ex}{1,087,562} & \raisebox{-0.5ex}{1,541,514} & \raisebox{-0.5ex}{1,087,562} & \raisebox{-0.5ex}{1,541,514} \\[0.5ex]

\raisebox{-0.5ex}{com-Youtube} & \raisebox{-0.5ex}{1,134,890} & \raisebox{-0.5ex}{2,987,624} & \raisebox{-0.5ex}{1,134,890} & \raisebox{-0.5ex}{2,987,624} \\[0.5ex]
\raisebox{-0.5ex}{roadNet-TX} & \raisebox{-0.5ex}{1,379,917} & \raisebox{-0.5ex}{1,921,660} & \raisebox{-0.5ex}{1,351,137} & \raisebox{-0.5ex}{1,879,201} \\[0.5ex]

\raisebox{-0.5ex}{roadNet-CA} & \raisebox{-0.5ex}{1,965,206} & \raisebox{-0.5ex}{2,766,607} & \raisebox{-0.5ex}{1,957,027} & \raisebox{-0.5ex}{2,760,388} \\[0.5ex]

\Xhline{2.5\arrayrulewidth}
\end{tabular}
}
\end{center}
\end{table}

In this section, we study the performance of our algorithms
by conducting
experiments on some classic
network models and real-world networks
taken from
KONECT~\cite{kunegis2013konect} and
SNAP~\cite{snapnets}.
We run our experiments on the largest components of these networks.
Related information of these networks
is shown in Table~\ref{SetNo},
where networks are shown in increasing order of their
numbers of vertices.


%

We implement our algorithms in Julia to facilitate
interactions with the SDDM solver contained in the
Laplacian.jl package\footnote{https://github.com/danspielman/Laplacians.jl}.
All of our experiments were run on
a Linux box with 4.2 GHz Intel i7-7700 CPU and 32G memory,
using a single thread.


\subsection{Accuracy of routine $\GainsEst$}

\begin{figure}[H]
\vspace{-11.5pt}
\centering
\includegraphics[width=0.7\textwidth]{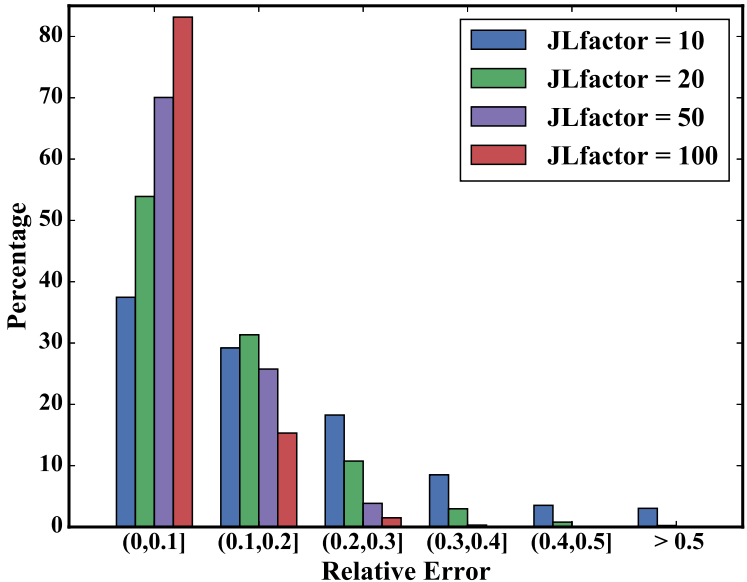}
\caption{Error distribution for routine $\GainsEst$ with
different JLfactors on network ca-HepTh.\label{fig:errors}}
\end{figure}

We first show the accuracy of routine $\GainsEst$.
We use quantity $\mathrm{JLfactor}$ to denote $q / \ln n$ in $\GainsEst$.
We run $\GainsEst$ with different $\mathrm{JLfactor}$s on network
arXiv High Energy Physics - Theory collaboration network (ca-HepTh).
And we set $S_i = \setof{u^*}$ where
$u^* \defeq \argmin\nolimits_{u\in V} 	\sum\nolimits_{v\in V}\er(u,v)$.
For each $\mathrm{JLfactor}$, we compute the relative error of gain
of each vertex $u\neq u^{*}$, given by
\[\frac{\abs{g_u - \kh{\ee_u^T \LL_{-u^*}^{-2} \ee_u} / \kh{\ee_u^T \LL_{-u^*}^{-1} \ee_u}}}
	{\kh{\ee_u^T \LL_{-u^*}^{-2} \ee_u} / \kh{\ee_u^T \LL_{-u^*}^{-1} \ee_u}}. \]
We then draw the distribution of relative errors with different $\mathrm{JLfactor}$s
in a histogram.
The results are shown in Figure~\ref{fig:errors}.
We observe that the errors become small when $\mathrm{JLfactor}$ gets larger.
Moreover, almost all errors are in the range $(0,0.5]$ when $\mathrm{JLfactor} \geq 20$.

We set $\mathrm{JLfactor} = 20$ in all other experiments.
We will show that this is sufficient for our greedy algorithm
to obtain good solutions empirically.

\begin{figure}[h]
\centering
\includegraphics[width=0.7\textwidth]{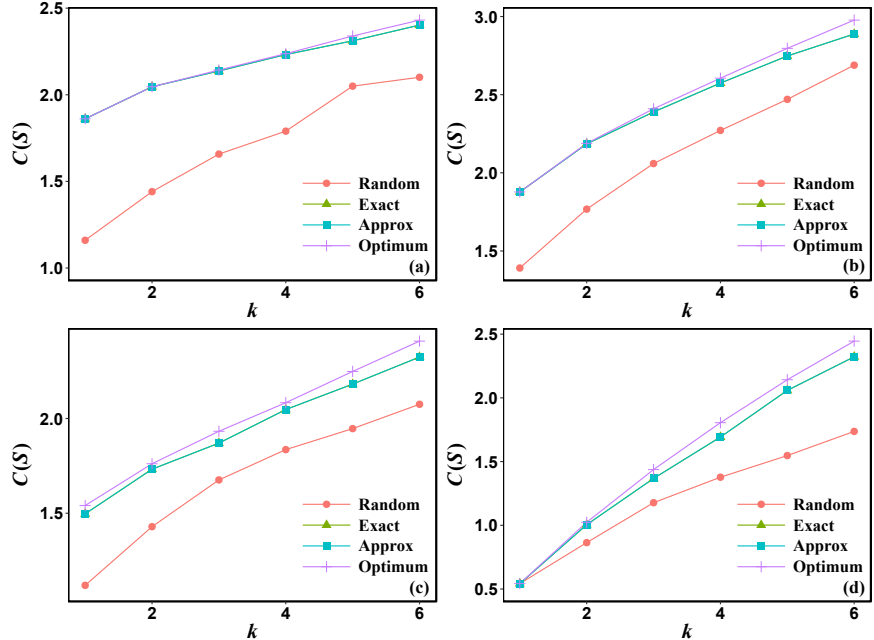}
\caption{Current flow closeness of
vertex sets returned by
$\exactGreedy$, $\FastGreedy$,
random and  optimum strategies  on four models:
BA (a), WS  (b), ER (c),
and a regular ring lattice (d).\label{ModelsOpt}}
\end{figure}

\subsection{Effectiveness of Greedy Algorithms}

\begin{figure}[h]
\centering
\includegraphics[width=0.7\textwidth]{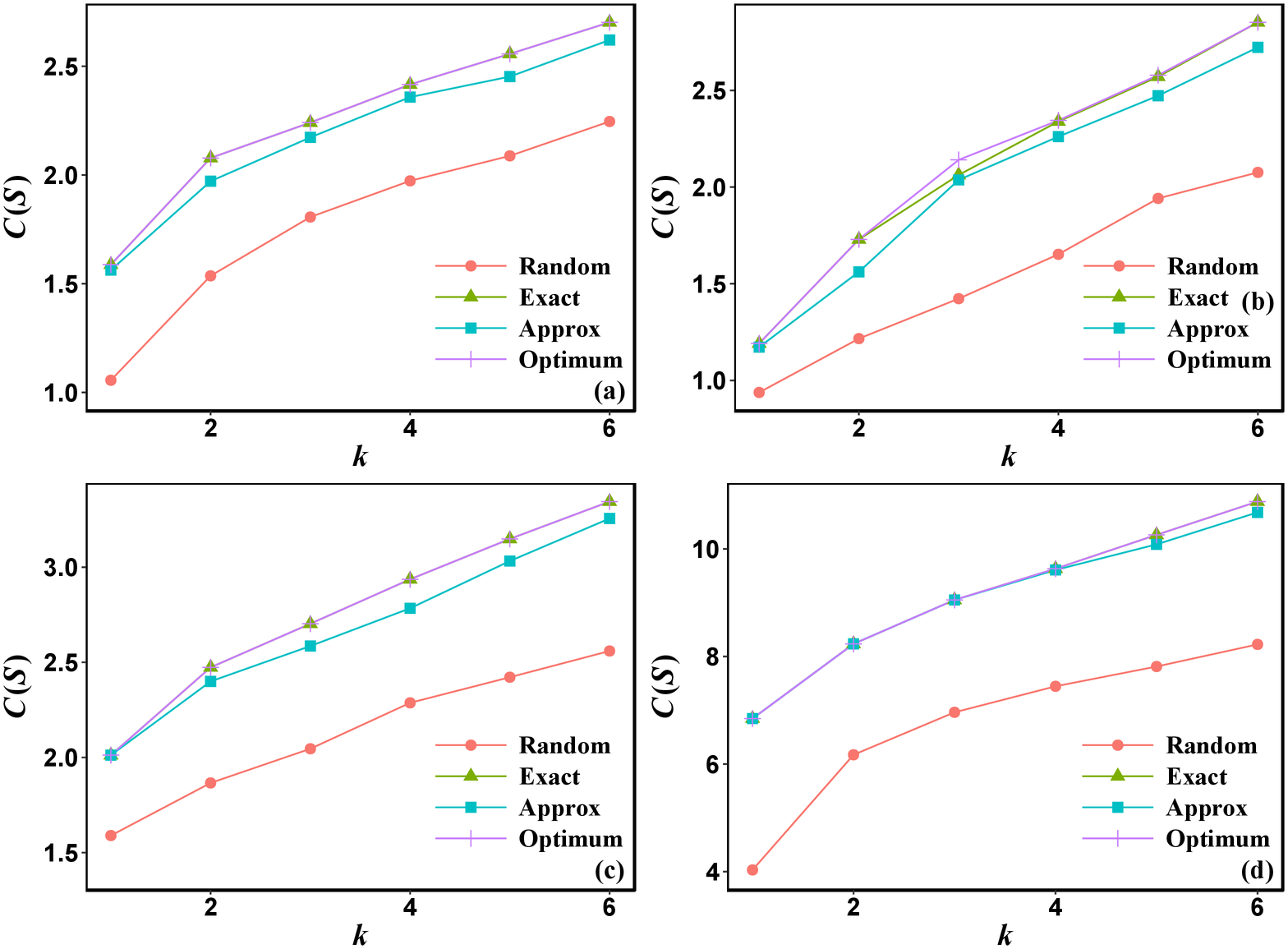}
\caption{Current flow closeness of
vertex sets returned by
$\exactGreedy$, $\FastGreedy$,
random and  optimum strategies  on four networks:
Dolphins (a), 	Contiguous USA (b), Zachary karate club (c),
and Windsurfers (d).\label{NetsOpt}}
\end{figure}

We show the effectiveness of our algorithms
by comparing the results of our algorithms with
the optimum solutions on four small model networks
(Barab\'asi-Albert (BA)~\cite{BaAl99}, Watts-Strogatz (WS)~\cite{WaSt98},
Erd{\"o}s-R{\'e}nyi (ER)~\cite{ER59}, and a regular ring lattice~\cite{WaSt98})
and four small realistic networks (Dolphins, Contiguous USA, Zachary karate club,
and Windsurfers).
We are able to compute the optimum solutions on these
networks because of their small sizes.

For each $k=1,2,\ldots,6$,
we first find the set of $k$ vertices
with the optimum current flow closeness by
brute-force search.
We then compute the current flow closeness of
the vertex sets
returned by $\exactGreedy$ and $\FastGreedy$.
Also, we compute the solutions returned by
the random scheme,
which chose $k$ vertices uniformly at random.
The results are shown in Figure~\ref{ModelsOpt} and~\ref{NetsOpt}.
We observe that
the current flow closenesses of
the solutions returned by our two greedy algorithms
and the optimum solution
are almost the same.
This means that the approximation ratios
of our greedy algorithms are  significantly  better
than their theoretical guarantees.
Moreover, the solutions returned by our greedy algorithms
are much better than those returned by the random scheme.


We then demonstrate the effectiveness of our
algorithms further by comparing it to the random scheme,
as well as two other schemes,
Top-degree and Top-cent,
on six larger networks.
Here the Top-degree scheme chooses $k$ vertices with
highest degrees, while
the Top-cent scheme chooses $k$ vertices with
highest current flow closeness centrality
according to Definition~\ref{def:ccsingle}.
A comparison of  results for these
five algorithms is shown in Figure~\ref{ComBase}.
We observe that our two greedy algorithms
obtain similar approximation ratios,
and both outperform the other three schemes
(random, Top-degree, and Top-cent).

\begin{figure}[h]
\centering
\includegraphics[width=.7\textwidth]{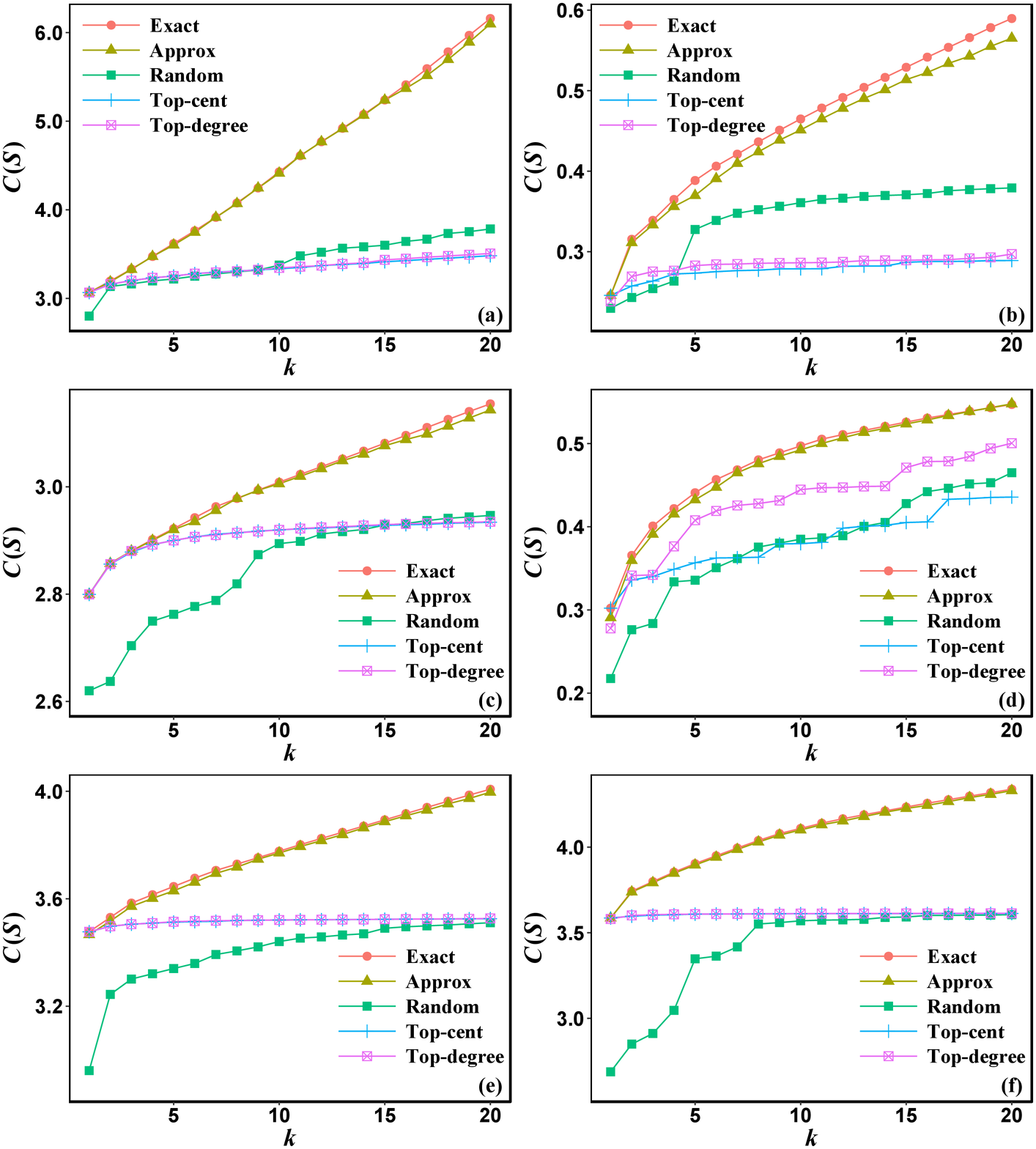}
\caption{Current flow closeness of vertex sets returned versus the number $k$ of vertices chosen for the five algorithms on David Copperfield (a),
Euroroad (b), U. Rovira i Virgili (c), US power grid (d), Hamsterster full (e), and Reactome (f).\label{ComBase} }
\end{figure}

\subsection{Efficiency of Greedy Algorithms}

We now show that
algorithm $\FastGreedy$
runs much faster than $\exactGreedy$,
especially on large-scale networks.
We run our two greedy algorithms
on a larger set of real-world networks.
For each network, we use both greedy algorithms
to choose $k=10$ vertices,
and then compare their running times
as well as the solutions they returned.
For networks with more than 30000 vertices,
we compute the current flow closeness of the vertex set
returned by $\FastGreedy$ by combining
Fast Laplacian solvers and Hutchinson's trace estimation~\cite{Hut89,AvTo11}.
We list the results in
Table~\ref{tab:performance_comparison}.
We observe that the ratio of the
running time of $\exactGreedy$ to that of
$\FastGreedy$ increases rapidly when
the network size becomes larger,
while the vertex sets they returned
have approximately the same current flow closeness.
Moreover, $\FastGreedy$ is able to solve
current flow closeness maximization for
networks with more than $10^6$ vertices
in a few hours,
whereas $\exactGreedy$ fails
because of its high time and space complexity.

%
%

\section{Conclusions}

In this paper, we extended the notion of current flow  closeness centrality (CFCC) to a group of vertices. For a vertex group $S$ in an $n$-vertex graph with $m$ edges, its CFCC $C(S)$ equals the ratio of $n$ to the sum of effective resistances between $S$ and all other vertices. We then considered the problem of finding the set $S^*$ of $k$ vertices with an aim to maximize $C(S^*)$, and solved it by considering an equivalent problem of minimizing $\trace{\LL_{-S}^{-1}}/n$. We showed that the problem is NP-hard, and proved that the objective function is monotone and supermodular. We devised two approximation algorithms for minimizing $\trace{\LL_{-S}^{-1}}/n$ by iteratively selected $k$ vertices in a greedy way. The first one achieves a $(1-\frac{k}{k-1}\cdot \frac{1}{e})$ approximation ratio in time $O(n^3)$; while the second one obtains a $(1-\frac{k}{k-1}\cdot\frac{1}{e}-\eps)$ approximation factor in time $\Otil (mk\eps^{-2})$. We conducted extensive experiments on model and realistic networks, the results of which show that both algorithms can often give  almost optimal solutions. In particular, our second algorithm is able to scale to huge networks, and quickly gives good approximate solutions in networks with more than $10^6$ vertices. In future works, we plan to introduce and study the betweenness group centrality based on current flow~\cite{Ne05}, which takes into account all possible paths.

\begin{table}[h]
\setlength{\abovecaptionskip}{5.pt}
\setlength{\belowcaptionskip}{-0.cm}
  \centering
  \fontsize{8}{8}\selectfont
  \begin{threeparttable}
  \caption{The average running times and results of
  $\exactGreedy$ (Exact) and
  $\FastGreedy$ (Approx) algorithms on a larger set of
real-world networks, as well as the ratios of running
times of $\exactGreedy$ to those of $\FastGreedy$,
and ratios of current-flow closenesses
achieved by $\FastGreedy$
to those achieved by $\exactGreedy$. \label{time}}
  \label{tab:performance_comparison}
    \begin{tabular}{ccccccc}
    \toprule
    \multirow{2}{*}{Network}&
    \multicolumn{3}{c}{Time (seconds)}&\multicolumn{3}{c}{Current flow closeness}\cr
    \cmidrule(lr){2-4} \cmidrule(lr){5-7}
    & Exact & Approx &Ratio&Exact &Approx &Ratio\cr
    \midrule
    Protein
    	& 1.35  & 1.34  & 1.01
    	& 0.7451 & 0.7388 & 0.9915 \cr
    	
    Vidal
    	& 8.04  & 3.40  & 2.36
      & 1.3179 &  1.3138  & 0.9968 \cr
    	
    ego-Facebook
    	& 13.26  & 2.56  & 5.17
      & 1.0195 &  1.0195  & 1.00 \cr

	 ca-Hepth
    	& 196.67  & 15.80  & 12.44
      & 1.5340 &  1.5267  & 0.9952\cr

	 PG-Privacy
    	& 367.71  & 17.25  & 21.31
      & 0.7155 &  0.7141  & 0.9980\cr


    CAIDA
    	& 5530.32  & 30.64  & 464.96
      & 1.3948 &  1.3945  & 0.9997\cr

    ego-Twitter          &- & 562.27       & -      & - & 5.5052      & -     \cr
    com-DBLP             &- & 1700.13       & -      & - & 1.6807      & -     \cr
    com-Youtube          &- & 6374.02       &-       & - & 1.0382      & -     \cr
    roadNet-PA           &- & 9802.54       &-       & - & 0.3089 & -     \cr
    roadNet-TX           &- & 13671.05       &-       & - & 0.2260      & -     \cr
    roadNet-CA           &- & 22853.20       &-       & - & 0.2630      & -     \cr
    \bottomrule
    \end{tabular}
    \end{threeparttable}

\end{table}

\newpage
\newcommand{\etalchar}[1]{$^{#1}$}


}

\end{document}